\documentclass[11pt]{article}
\usepackage{fullpage}
\usepackage{amsmath}
\usepackage{amssymb}
\usepackage{amsthm}
\usepackage{fancybox}
\usepackage{graphicx}
\usepackage{caption}
\usepackage{subcaption}
\usepackage{float}
\newtheorem{theorem}{Theorem}[section]

\newtheorem{lemma}[theorem]{Lemma}
\newtheorem{proposition}[theorem]{Proposition}

\newtheorem{definition}[theorem]{Definition}
\newtheorem{remark}[theorem]{Remark}
\newtheorem{warning}[theorem]{Warning}
\newcommand{\CAT}{\textrm{CAT}}
\newcommand{\BHV}{\textrm{BHV}}
\newcommand{\aC}{\mathcal{C}}
\newcommand{\aD}{\mathcal{D}}
\newcommand{\aE}{\mathcal{E}}
\newcommand{\aK}{\mathcal{K}}
\newcommand{\aP}{\mathcal{P}}

\newcommand{\depth}{\textrm{depth}}

\begin{document}

\title{Genomic data analysis in tree spaces}

\author{Sakellarios Zairis$^{1}$, Hossein Khiabanian$^{1,\dagger}$, Andrew J. Blumberg$^{2}$, and Raul Rabadan$^{1}$\\
\\
$^1$ Department of Systems Biology, Columbia University\\
$^2$ Department of Mathematics, UT Austin\\
$^\dagger$ Present address: Rutgers Cancer Institute, Rutgers University\\}
\maketitle


\begin{abstract}

Phylogenetic trees are arguably the most common representation of evolutionary processes; they have been used to characterize pathogen spread, the relationship between different species, and the evolution of cancers.
Comparison between different trees is a key part of the analysis of evolutionary phenomena in this framework.
For instance, one might compare the evolutionary trajectories of tumors in different patients to study the differential response to therapy.

Recently, an elegant approach has been introduced by Billera-Holmes-Vogtmann that allows a systematic comparison of different evolutionary histories using the metric geometry of tree spaces.
We begin by reviewing in detail the relevant mathematical and computational foundations for applying standard techniques from machine learning and statistical inference in these spaces, which we refer to as evolutionary moduli spaces.

In many problem settings one encounters heavily populated phylogenetic trees, where the large number of leaves encumbers visualization and analysis in the relevant evolutionary moduli spaces.
To address this issue, we introduce {\em tree dimensionality reduction}, a structured approach to reducing large and complex phylogenetic trees to a distribution of smaller trees.
We prove a stability theorem ensuring that small perturbations of the large trees are taken to small perturbations of the resulting distributions.

We then present a series of four biologically motivated applications to the analysis of genomic data, spanning cancer and infectious disease.
The first quantifies how chemotherapy can disrupt the evolution of common leukemias.
The second examines a link between geometric information and the histologic grade in relapsed gliomas, where longer relapse branches were specific to high grade glioma.
The third concerns genetic stability of xenograft models of cancer, where heterogeneity at the single cell level increased with later mouse passages.
The last studies genetic diversity in seasonal influenza A virus.
We apply tree dimensionality reduction to project 24 years of longitudinally collected H3N2 hemagglutinin sequences into distributions of smaller trees spanning between three and five seasons.
A negative correlation is observed between the influenza vaccine effectiveness during a season and the variance of the distributions produced using preceding seasons' sequence data.
We also show how tree distributions relate to antigenic clusters and choice of influenza vaccine.
These results provide compelling evidence that our formalism exposes links between genomic data of influenza A and important clinical observables, namely vaccine selection and efficacy.

\end{abstract}
\newpage
\tableofcontents


\section{Introduction}

The importance of phylogenetic tree structures in biological sciences cannot be overstated.
In 1859, Charles Darwin proposed the tree as a metaphor for the process of species generation through branching of ancestral lineages~\cite{darwin1859origin}.
Since then, tree structures have been used pervasively in biology, where terminal branches can represent a diverse set of biological entities and taxonomic units including individual or families of genes, organisms, populations of related organisms, species, genera, or higher taxa.

Although the widespread use of trees has been recently questioned on the grounds that some biological taxa are not the results of simple branching processes~\cite{doolittle1999phylogenetic, chan2013topology}, a large set of evolutionary processes are well captured by tree-like structures.
In particular, trees describe clonal evolution events that start from asexual reproduction of a single organism (the primordial clone), which mutates and differentiates into a large progeny (see Figure~\ref{fig:illustration_1})~\cite{khiabanian2014viral}.
Examples of these processes include single gene phylogeny in non-recombinant viruses, evolution of bacteria that are not involved in horizontal gene transfer events, and metazoan development from a single germ cell.
Recent developments in genomics allow the study of such events in exquisite detail, particularly at single cell/clone levels~\cite{navin2011tumour, shalek2013single, eirew2014dynamics}.

An important example of clonal evolution can be observed in cancer where a single cell replicates and spreads uncontrollably.
The answers to many clinical and biological questions involve investigating different phases of a tumor's clonal evolution.
How do tumors originate?
How do tumors spread?
How might a particular therapy disrupt tumor evolution?
How does evolutionary information correlate with disease prognosis?
How can patients' tumors be classified according to their evolutionary history?
It is also essential to assess the value of genomic information measured at a particular time for predicting subsequent events in tumor progression.
If the evolution of the tumor is well approximated by a simple stepwise accumulation of alterations in the dominant clone, future observations are expected to contain all the mutational information at present.
This scenario is often termed linear evolution, in reference to the shape of the implied phylogeny.
On the other hand, if an ancestral clone gains resistance to therapy and dominates the long-term population, a very different shape of tree would be observed, often termed branched evolution to distinguish from the linear case.

These and many other questions associated with clonal processes can be expressed in terms of comparing evolutionary histories.
For instance, stratifying cancer patients according to their evolutionary history requires a way of comparing trees and associating summary statistics to clouds of trees.
In a recent paper~\cite{zairis2014moduli}, we proposed applying the metric geometry of the space of phylogenetic trees, constructed by Billera-Holmes-Vogtmann, to this problem.
As illustrated in Figure~\ref{fig:illustration_2}, each history/patient is represented by a point in a space of trees.
We referred to these tree spaces as ``evolutionary moduli spaces,'' in which classifying histories translates into finding patterns within a point cloud.

In this paper, we develop these applications in detail and introduce new machinery for dimensionality reduction.
We begin by reviewing (in Section~\ref{sec:phylospace}) some of the fundamental notions regarding the geometry of evolutionary moduli spaces, following Billera-Holmes-Vogtman~\cite{billera2001geometry}, Sturm~\cite{sturm2003probability}, and our prior treatment~\cite{zairis2014moduli}.
Motivated by several biological applications, we also introduce a projective version of these spaces and study some of their geometric properties.
In Section~\ref{sec:ML}, we give an overview of the application of standard machine learning and statistical inference techniques in these spaces.
Applying these techniques, in Section~\ref{sec:cancer} we describe the interpretive power of evolutionary moduli spaces for studying cancer genomic data in three examples: evolution of chronic lymphocytic leukemia under therapy (Section~\ref{sec:CLL}), grade at relapse in glioma (Section~\ref{sec:glioma}), and clonal dynamics in patient-derived xenografts (Section~\ref{sec:xeno}).
In each case, we find that representing the problem in terms of the evolutionary moduli spaces permits meaningful inference with clinical significance.

The trees arising from clonal evolution in cancers tend to be quite small, since each leaf corresponds to a tissue sample from the tumor.
However, when studying the evolution of a viral outbreak with thousands of isolates, or the genomic characterization of thousands of single cells, we expect very large trees.
Visualization and analysis of large phylogenetic trees is problematic.
In Section~\ref{sec:treedimred}, we introduce a procedure for ``tree dimensionality reduction'' as a means to visualize and study densely populated trees by projecting them onto distributions of smaller trees.
We argue that the properties of these distributions capture salient information about the large trees.
To support this contention, we provide both theoretical and experimental validation.
We prove a ``stability'' theorem (Theorem~\ref{thm:stab}) which bounds perturbation in the projected trees in terms of perturbation of the large tree; this implies in particular that the procedure is robust to certain kinds of noise.
We also describe the use of tree dimensionality reduction to produce a novel genomic predictor for influenza vaccine effectiveness in Section~\ref{sec:flu}.

\begin{figure}
    \begin{subfigure}{0.5\linewidth}
    \centering
    \includegraphics[height=3in]{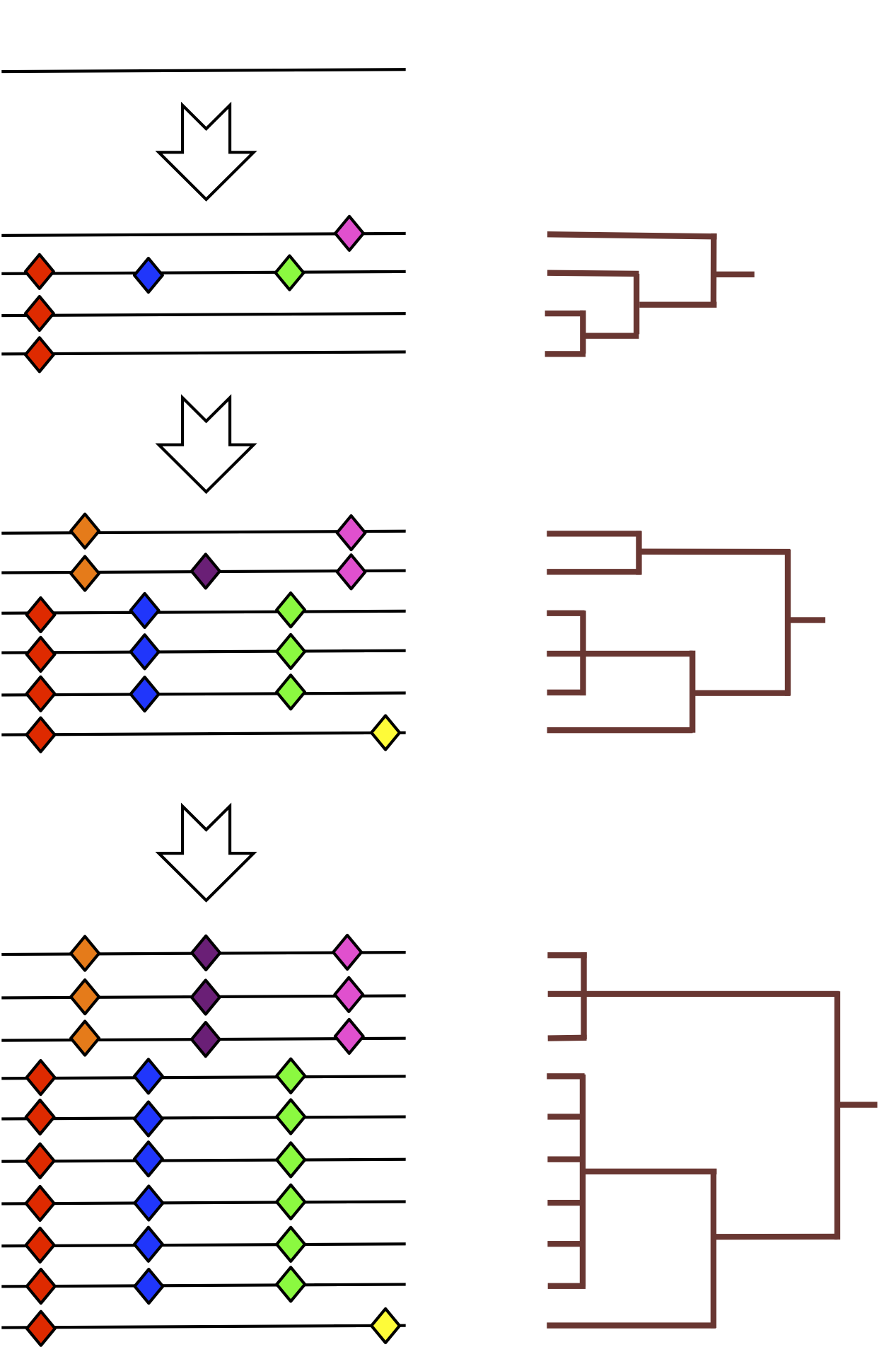}
    \end{subfigure}
    ~
    \begin{subfigure}{0.5\linewidth}
    \centering
    \includegraphics[height=3in]{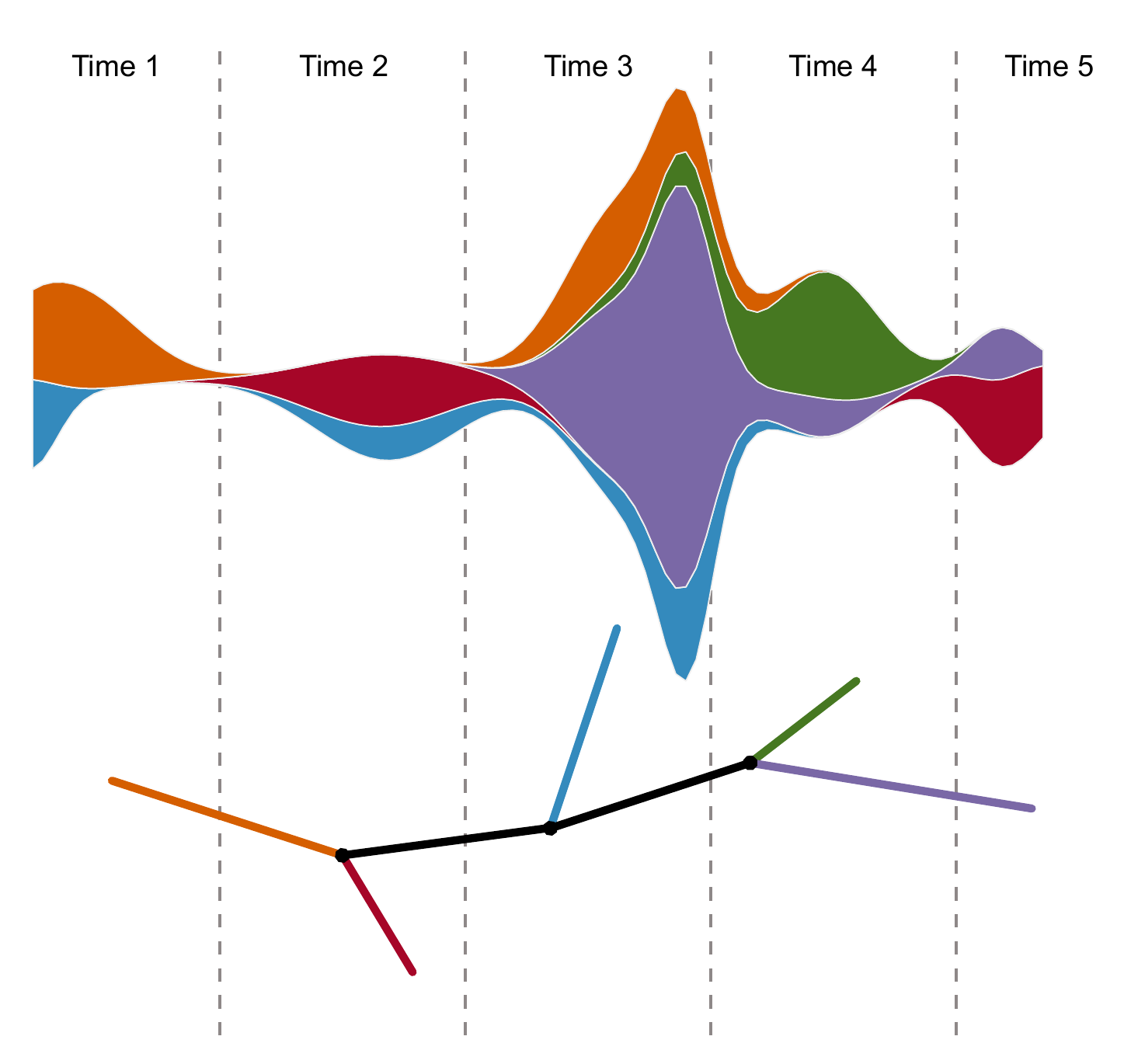}
    \end{subfigure}
    \caption{{\bf Clonal evolution of an asexually reproducing genome}. (Left) Through acquisition of mutations, the primordial clone gives rise to a large heterogeneous population, whose evolutionary history can be accurately described by a tree. (Right) Longitudinal sampling of a clonal population permits the construction of phylogenetic trees that approximate the underlying history. Subpopulations are represented in different colors; random sampling of a particular genotype at each time point is illustrated in the color of the external branch in the tree. This tree is one of many that could be observed when sampling this population.}
     \label{fig:illustration_1}
\end{figure}

\begin{figure}
    \centering
    \includegraphics[height=3.5in]{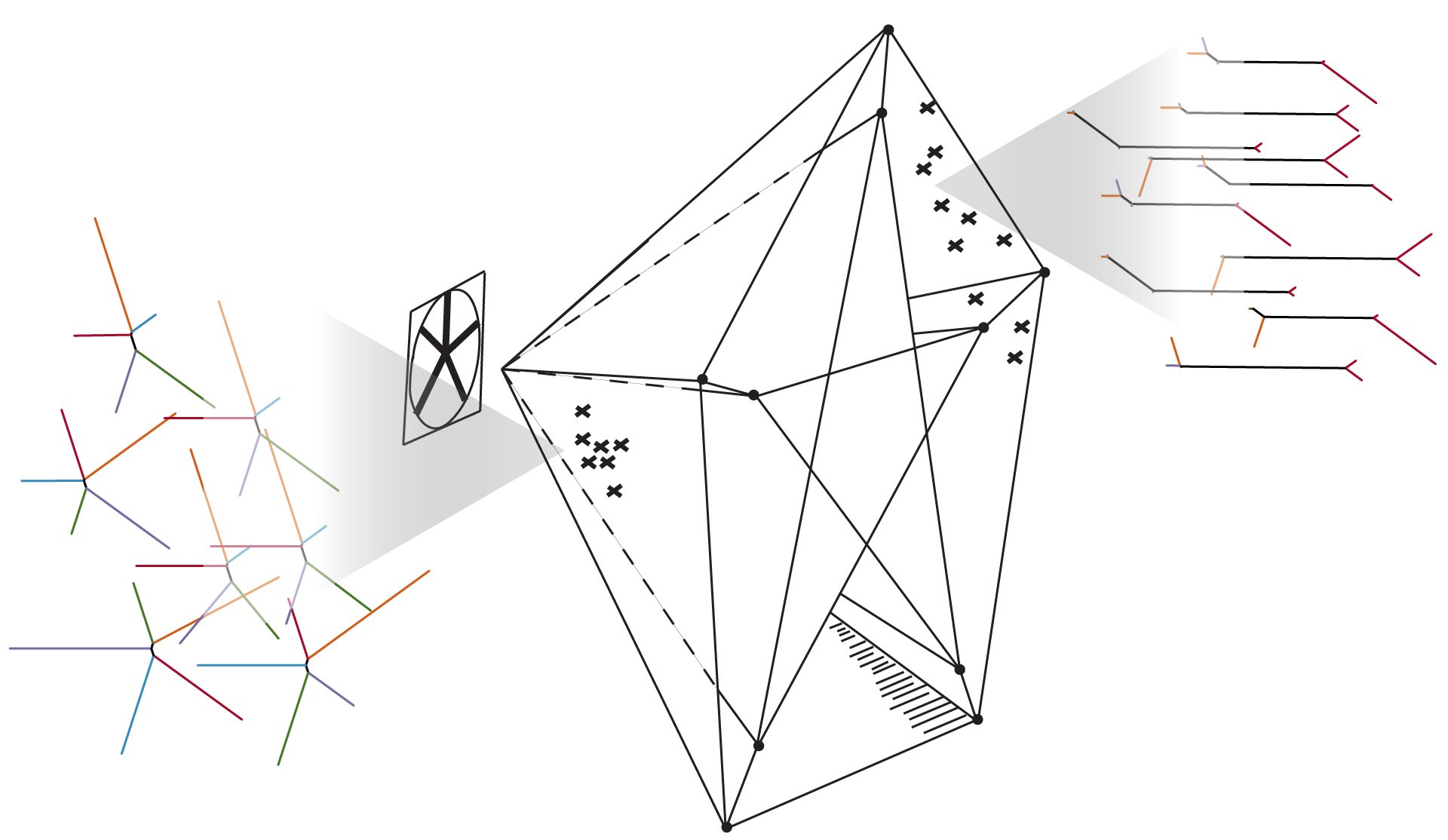}
    \caption{{\bf Moduli space of phylogenetic trees describing clonal evolution.} Collections of trees are points in a metric space, forming a point cloud. Trees with the same topology live in the same orthant, and crossing into an adjacent orthant corresponds to a tree rotation. Points closer to the vertex of the cone have relatively little internal branch length, while points near the base of the cone have little weight in the external branches.}
    \label{fig:illustration_2}
\end{figure}


\section{Spaces of phylogenetic trees}\label{sec:phylospace}

The foundation of our framework for analysis is the {\em metric geometry} of the space of phylogenetic trees~\cite{billera2001geometry}.
The purpose of this section is to review in detail the spaces of phylogenetic trees that we work with and their geometric structure.
We begin with a rapid review of the geometry of geodesic metric spaces and the theory of cubical complexes.
We then review the definition and properties of the Billera-Holmes-Vogtmann metric space of phylogenetic trees and its metric geometry, following the the excellent original treatment.
Finally, we discuss the properties of projective versions of tree spaces which are relevant for some of the biological applications.

\subsection{A rapid review of metric geometry}

In this subsection we quickly explain the foundations of metric geometry.
See~\cite{bridson99, burago01} for comprehensive textbooks on the subject.
A metric space $(X,d)$ is a set $X$ equipped with a distance function $d \colon X \times X \to \mathbb{R}^{\geq 0}$ having the properties that $d(x,y) = d(y,x)$, $d(x,y) = 0$ if and only if $x = y$, and $d(x,z) \leq d(x,y) + d(y,z)$ for all $x,y,z \in X$ (the triangle inequality).
Although metric spaces often arise in contexts in which there is not an evident notion of geometric structure, it turns out that under very mild hypotheses a metric space $(X,d)$ can be endowed with structures analogous to those arising on Riemannian manifolds.
A metric space is a length space if the distance $d(x,y)$ is realized as the infimum of the lengths of paths joining $x$ and $y$.
A length space $(X,d)$ is a geodesic metric space if any two points $x$ and $y$ can be joined by a path with length precisely $d(x,y)$.
A key insight of Alexandrov is that {\em curvature} makes sense in any geodesic metric space~\cite{alexandrov1957uber}.

The idea is that the curvature of a space can be detected by considering the behavior of the area of triangles, and triangles can be defined in any geodesic metric spaces.
Specifically, given points $p, q, r$, we have the triangle $T = [p,q,r]$ with edges the paths that realize the distances $d(p,q)$, $d(p,r)$, and $d(q,r)$.
The connection between curvature and area of triangles comes from the observation that given side lengths $(\ell_1, \ell_2, \ell_3) \subset \mathbb{R}^3$, a triangle with these side lengths on the surface of the Earth is ``fatter'' than the corresponding triangle on a Euclidean plane.
To be precise, we consider the distance from a vertex of the triangle to a point $p$ on the opposite side --- in a fat triangle, this distance will be larger than in the the corresponding Euclidean triangle (and smaller in a thin triangle).

Given a triangle $T=[p,q,r]$ in $(X,d)$, we can find a corresponding triangle $\tilde{T}$ in Euclidean space with the same edge lengths.
Given a point $z$ on the edge $[p,q]$, a comparison point in $\tilde{T}$ is a point $\tilde{z}$ on the corresponding edge $[\tilde{p}, \tilde{q}]$ such that $d_E(\tilde{z}, \tilde{p}) = d(z,p)$ and $d_E(\tilde{z},\tilde{q}) = d(z,q)$.
(Where here $d_E$ denotes the Euclidean metric.)
We say that a triangle $T$ in $M$ satisfies the $\CAT(0)$ inequality if for every such pair $(z, \tilde{z})$, we have $d(r,z) \leq d_E(\tilde{r},\tilde{z})$.
If every triangle in $M$ satisfies the $\CAT(0)$ inequality then we say that $M$ is a $\CAT(0)$ space.

More generally, let $M_{\kappa}$ denote the unique two-dimensional Riemannian manifold with curvature $\kappa$.
The diameter of $M_{\kappa}$ will be denoted $D_{\kappa}$.
Then we say that a geodesic metric space $M$ is $\CAT(\kappa)$ if every triangle in $M$ with perimeter $\leq 2D_{\kappa}$ satisfies the inequality above for the corresponding comparison triangle in $M_{\kappa}$.
If $\kappa' \leq \kappa$, any $CAT(\kappa')$ space is also $\CAT(\kappa)$.
A $n$-dimensional Riemannian manifold $M$ that is sufficiently smooth has sectional curvature $\leq \kappa$ if and only if $M$ (regarded as a metric space) is $\CAT(\kappa)$.
For example, Euclidean spaces are $\CAT(0)$, spheres are $\CAT(1)$, and hyperbolic spaces are $\CAT(-1)$.

As described, $\CAT(\kappa)$ is a global condition; we will say that a metric space $(X,d)$ is locally $\CAT(\kappa)$ if for every $x$ there exists a radius $r_x$ such that $B_{r_x}(x) \subseteq X$ is $\CAT(\kappa)$.
For example, the flat torus (obtained by identifying opposite edges in a rectangle) is locally $\CAT(0)$ but not globally $\CAT(0)$.
The Cartan-Hadamard theorem implies that a simply-connected metric space that is locally $\CAT(0)$ is also globally $\CAT(0)$.

A remarkably productive observation of Gromov is that many geometric properties of Riemannian manifolds are shared by $\CAT(\kappa)$ spaces.
In particular, $\CAT(\kappa)$ spaces with $\kappa \leq 0$ (referred to as {\em non-positively curved metric spaces}) admit unique geodesics joining each pair of points $x$ and $y$, balls $B_{\epsilon}(x)$ are convex and contractible for all $x$ and $\epsilon \geq 0$, and midpoints of geodesics are well-behaved.
As a consequence, there exist well-defined notions of mean and variance of a set of points, and more generally one can develop some of the foundations of classical statistics, as we review below in Section~\ref{sec:ML}.

\subsection{Cubical complexes and their links}

In this subsection, we review the theory of cubical complexes, which provide a rich source of examples of $\CAT(0)$ metric spaces (again, see~\cite{burago01} or~\cite{bridson99} for textbook treatments).
It is in general very difficult to determine for an arbitrary metric space whether it is $\CAT(\kappa)$ for any given $\kappa$.
Even for finite polyhedra where the metric is induced from the Euclidean metric on each face, this problem does not have a general solution.
The important of cubical complexes in this context comes from an effective criterion for determining if they are non-positively curved (i.e., $\CAT(0)$).

Let $I^n \subseteq \mathbb{R}^n$ denote the $n$-dimensional unit cube $[0,1] \times \ldots \times [0,1]$, regarded as inheriting a metric structure from the standard metric on $\mathbb{R}^n$.
A codimension $k$ face of the cube $I^n$ is determined by fixing $k$ coordinates to be in the set $\{0,1\}$.
A cubical complex is a metric space obtained by gluing together cubes via the data of isometries of faces, subject to the condition that two cubes are connected by at most a single face identification and no cube is glued to itself.
The metric structure is the length metric induced from the Euclidean metric on the cubes, i.e., the distance between $x$ and $y$ is the infimum of the lengths over all paths from $x$ to $y$ that can be expressed as the union of finitely many segments each contained within a cube.
When the cubical complex $C$ is finite or locally finite, results of Bridson~\cite{bridson91} and Moussong~\cite{moussong} imply that $C$ is a complete geodesic metric space.

Gromov gave a criterion for a cubical complex to be $\CAT(0)$ that is often possible to check in practice.
In order to explain this criterion, we need to review the notion of the link of a vertex in a cubical complex.

Fix a vertex $v$ in a cubical complex $C$ and a cube $C_i \cong I^m \subseteq C$ such that $v$ is a vertex of $C_i$.
For fixed $\epsilon > 0$, the all-right spherical simplex associated to $(C_i,v)$ is the subset 
\[
S(C_i,v) = \{z \in C_i \, | \, d(z,v) = \epsilon\}.
\]
The set $S(C_i,v)$ has a metric induced by the Euclidean angle metric.
The faces of $S(C_i,v)$ are defined as the intersections of $S(C_i,v)$ with faces of $C_i$; equivalently, these are the all-right spherical simplexes associated to faces of $C_i$.
The collection of all-right spherical simplices for all pairs $(C_i, v)$ forms a polyhedral complex with metric given by the length metric induced from the angle metrics; this is referred to as a spherical complex.
Forgetting the metric structure, the all-right spherical simplices also form an abstract simplicial complex.
(Recall that an abstract simplicial complex is simply a set of subsets of a set $V$ that is closed under passage to subsets.)

The link of a vertex $v$ in a cubical complex $C$ is the spherical complex obtained as the subset  
\[
L(v) = \{z \in C \, | \, d(z,v) = \epsilon\},
\]
for fixed $0 < \epsilon < 1$.
Gromov's criterion now states that the cubical complex $C$ is locally $\CAT(0)$ if and only if the link is $\CAT(1)$ or the abstract simplicial complex underlying the link is flag.
(Recall that a flag complex is a simplicial complex in which a $k$-simplex is in the complex if and only its $1$-dimensional faces are in the complex.)

As an easy application of Gromov's criterion, we conclude the section by showing the standard result that the Cartesian product of locally $\CAT(0)$ cubical complexes is itself a locally $\CAT(0)$ cubical complex.
Let $X$ and $Y$ be cubical complexes that are $\CAT(0)$.
Since $I^n \times I^m \cong I^{n+m}$, it is clear that $X \times Y$ has the structure of a cubical complex.
The set of vertices of $X \times Y$ is given by the product of the sets of vertices of $X$ and $Y$ respectively.
The link of a vertex $(v,v')$ in $X \times Y$, regarded as an abstract simplicial complex, is the join of $L(v)$ and $L(v')$, which we denote $L(v) \ast L(v')$ (the join of complexes $S_1$ and $S_2$ is obtained by considering all pairwise unions of elements of $S_1$ and $S_2$).
Finally, since the join of flag complexes is easily seen to be a flag complex, Gromov's criterion now implies that $X \times Y$ is locally $\CAT(0)$.

\subsection{The Billera-Holmes-Vogtmann spaces of phylogenetic trees}

A phylogenetic tree with $m$ leaves is a weighted, connected graph with no cycles, having $m$ distinguished vertices (referred to as {\em leaves}) of degree $1$ and labeled $\{1, \ldots, m\}$.
All the other vertices are of degree $\geq 3$.
We refer to edges that terminate in leaves as {\em external} edges and the remaining edges are {\em internal}.

The space $\BHV_m$ of isometry classes of rooted phylogenetic trees with $m$-labelled leaves where the nonzero weights are on the internal branches was introduced and studied by Billera, Holmes, and Vogtmann~\cite{billera2001geometry}.
The space $\BHV_m$ is constructed by gluing together $(2m-3)!!$ positive orthants $\mathbb{R}^m_{\geq 0}$; each orthant corresponds to a particular tree
topology, with the coordinates specifying the lengths of the edges.
A point in the interior of an orthant represents a binary tree; if any of the coordinates are $0$, the tree is obtained from a binary tree by collapsing some of the edges.
We glue orthants together such that a (non-binary) tree is on the boundary between two orthants when it can be obtained by collapsing edges from either tree geometry.
Put another way, two tree topologies are adjacent when they are connected by a {\em rotation}, i.e., one topology can be generated from the other by collapsing an edge to length $0$ and then expanding out another edge from the incident vertex.

The metric on $\BHV_{m}$ is induced from the standard Euclidean distance on each of the orthants.
For two trees $t_1$ and $t_2$ which are both in a given orthant, the distance $d_{\BHV_{m}}(t_1,t_2)$ is defined to be the Euclidean distance between the points specified by the weights on the edges.
For two trees which are in different quadrants, there exist (many) paths connecting them which consist of a finite number of straight lines in each quadrant.
The length of such a path is the sum of the lengths of these lines, and the distance $d_{\BHV_{m}}(t_1,t_2)$ is then the minimum length over all such
paths.
For many points, the shortest path goes through the ``cone point'', the star tree in which all internal edges are zero.

Allowing potentially nonzero weights for the $m$ external leaves corresponds to taking the cartesian product with an $m$-dimensional orthant.
We will focus on the space \[\Sigma_m = \BHV_{m-1} \times \mathbb{R}^m_{\geq 0},\] which we refer as the evolutionary moduli space (the $m-1$ index
arises from the fact that we consider unrooted trees.)
There is a metric on $\Sigma_m$ induced from the metric on $\BHV_{m-1}$.
Specifically, for a tree $t$, let $t(i)$ denote the length of the external edge associated to the vertex $i$.
Then \[d_{\Sigma_m}(t_1,t_2) =
\sqrt{\left(d_{\BHV_{m-1}}(\bar{t}_1,\bar{t}_2)\right)^2 + \sum_{i=1}^m (t_1(i) -
  t_2(i))^2},\] where $\bar{t}_i$ denotes the tree in $\BHV_{m-1}$
obtained by forgetting the lengths of the external edges (e.g.,
see~\cite{owen2011fast}).
As explained in ~\cite[\S4.2]{billera2001geometry}, efficiently computing the metric on $\Sigma_m$ is a nontrivial problem, although there exists a polynomial-time algorithm~\cite{owen2011fast}.

The main result of Billera, Holmes, and Vogtmann is that the length metric on $\BHV_n$ endows this space with a (global) $\CAT(0)$ structure.
By subdividing each orthant into cubes in the evident fashion, $\Sigma_m$ is naturally a cubical complex where the metric we have described is the one induced from the Euclidean metric on the cubes; a straightforward combinatorial analysis of the link of $\Sigma_m$ implies the result via Gromov's criterion.
In addition, $\Sigma_m$ is clearly a complete and separable metric space; any tree can be approximated by a sequence of trees in the same orthant that have rational edge lengths.

\subsection{The projective evolutionary moduli space}

In evolutionary applications, we are often interested in classifying and comparing distinct behaviors by understanding the relative lengths of edges: rescaling edge lengths should not change the relationship between the branches~\cite{zairis2014moduli}.
Motivated by this consideration, we define $\mathbb{P} \Sigma_m$ to be the subspace of $\Sigma_m$ consisting of the points $\{t_i\}$ in each orthant for which the constraint $\sum_{i} t_i = 1$ holds.

We denote the space of trees with internal edges of fixed length by $\tau_{m-1}$.
The space of $m$ external branches whose lengths sum to $1$ is the standard $m-1$ dimensional simplex $\Delta_{m-1}$ in $\mathbb{R}^m$.
The constraint that the length of internal branches plus the external branches sum to $1$ implies that 
\[\mathbb{P}\Sigma_m = \tau_{m-1} \star \Delta_{m},\]
where here $\star$ denotes the join of two spaces, using the Milnor model of the join.
We can also describe $\mathbb{P}\Sigma_m$ as the link on the origin in $\Sigma_m$.

There are various possible natural metrics to consider on $\mathbb{P} \Sigma_m$.
The simplest way to endow $\mathbb{P} \Sigma_m$ with a metric is to use the induced intrinsic metric specified by paths in $\Sigma_m$ constrained to lie entirely within $\mathbb{P} \Sigma_m$.
From the perspective of metric geometry, the characterization of $\mathbb{P} \Sigma_m$ as the link of the origin endows it with a ``spherical'' metric, and Gromov's criteria imply that with this metric, $\mathbb{P} \Sigma_m$ is a $\CAT(1)$ space.
(Alternatively, $\mathbb{P} \Sigma_m$ is the spherical join of $\tau_{m-1}$ and the spherical realization of the $\Delta_{m}$; since $\tau_{m-1}$ and $\Delta_{m}$ are $\CAT(1)$, so is their spherical join~\cite[II.3.15]{bridson99}.)
The theory of polyhedral complexes implies that in either case $\mathbb{P} \Sigma_m$ is a complete geodesic metric space~\cite[I.7.19]{bridson99}, and it is evidently separable.

Moreover, with the induced intrinsic metric $\mathbb{P} \Sigma_m$ is in fact a $\CAT(0)$ space; although $\tau_{m-1}$ has points which are not connected by unique geodesics (see Section~\ref{sec:inj} below for a more detailed discussion), the join with $\Delta_{m}$ introduces a new ``cone direction'' that changes the geometry.

\begin{theorem}
The projective moduli space $\mathbb{P} \Sigma_m$ endowed with the intrinsic metric is a $\CAT(0)$ space.
\end{theorem}

\begin{proof}
First, recall that $\tau_{m-1} \star \Delta_{k}$ is isomorphic to 
\[
\tau_{m-1} \star \underbrace{\Delta_0 \star \Delta_0 \ldots \star \Delta_0}_{k}.
\]
To see this, observe that a point in the join of $\tau_{m-1} \star \Delta_{k-1}$ with $\Delta_0$ can be described as a tuple 
\[
\left((wt_0, \ldots, wt_n), (wx_0, \ldots, wx_{k-1}), 1-w\right),
\]
where $\sum_{i=0}^n t_i + \sum_{i=0}^{k-1} x_i = 1$ and $w \in [0,1]$.  
This data is clearly equivalent to a tuple 
\[
\left((t_0, \ldots, t_n), (x_0, \ldots, x_{k-1}, x_{k})\right)
\]
where $\sum_{i=0}^n t_i + \sum_{i=0}^k x_i = 1$.

The fact that $\BHV_{m-1}$ is $\CAT(0)$ implies that the cone $\tau_{m-1} \star \Delta_0$ is $\CAT(0)$, and from this it follows by induction that $\tau_{m-1} \star \Delta_{m}$ is also $\CAT(0)$.
\end{proof}

To compute the intrinsic metric on $\mathbb{P}\Sigma_m$, we use $\epsilon$-nets and a local-to-global construction.
Recall that a set of points $S$ in a metric space $(X,\partial)$ is an $\epsilon$-net if for every $z \in X$, there exists $q \in S$ such that $\partial(z,x) < \epsilon$.
For a compact metric space equipped with a probability measure such that all non-empty balls in the metric space have nonzero measure, we can produce an $\epsilon$-net by sampling.
More precisely, it is straightforward to show that given a finite collection of measurable sets $\{A_1, A_2, \ldots, A_k\}$ and a probability measure $\mu$ on $\cup_i A_i$ such that $\mu(A_i) \geq \alpha > 0$, then given at least
\[
\frac{1}{\alpha}\left(\log k + \log(\frac{1}{\delta})\right)
\]
samples, with probability $1-\delta$ there is at least one sample in every $A_i$~\cite[5.1]{niyogi2008finding}.

Next, suppose that we have a metric space $(X,\partial)$ where there exists a constant $\kappa$ such that if $\partial(x,y) < \kappa$, it is easy to compute $\partial(x,y)$.
An algorithm for approximating $\partial$ on all of $X$ is then to take a dense sample $S \subset X$, form the graph $G$ with vertices the points of $S$ and edges between $x$ and $y$ when $\partial(x,y) < \kappa$, and define the distance between $x$ and $y$ in $X$ to be the graph metric on $G$ between the nearest points to $x$ and $y$ in $S$.
This distance can be efficiently computed using Dijkstra's algorithm~\cite{dijkstra1959note}.

When $S$ is an $\epsilon$-net for sufficiently small $\epsilon$ relative to $\kappa$, we can describe the quality of the resulting approximation to $\partial$~\cite[Thm. 2]{bernstein2000graph}.
In particular, if $\epsilon < \frac{\kappa}{4}$, then 
\[
\partial(x,y) \leq \partial_G(x,y) \leq (1 + 4\frac{\delta}{\epsilon}) \partial(x,y).
\]

Putting this all together, to approximate the metric on $\mathbb{P}\Sigma_m$ we take the union of $\epsilon$-nets on all of the simplices (including the faces) and form a $\kappa$-approximation $\partial_G$ as above.
In practice, the required density of samples is determined by looking at when the approximation converges (i.e., when the change in distances drops below a specified precision bound).
We sample densely on each simplex representing a tree topology on $m$ leaves, and explicitly include certain key singular points of the projective space.
Figure~\ref{fig:epsilon_net} contains two visualizations of the projective space, using force-directed layouts of the $\epsilon$-nets constructed on 4--leaved and 5--leaved trees respectively.

\begin{figure}
    \begin{subfigure}{0.5\linewidth}
    \centering
    \includegraphics[height=2in]{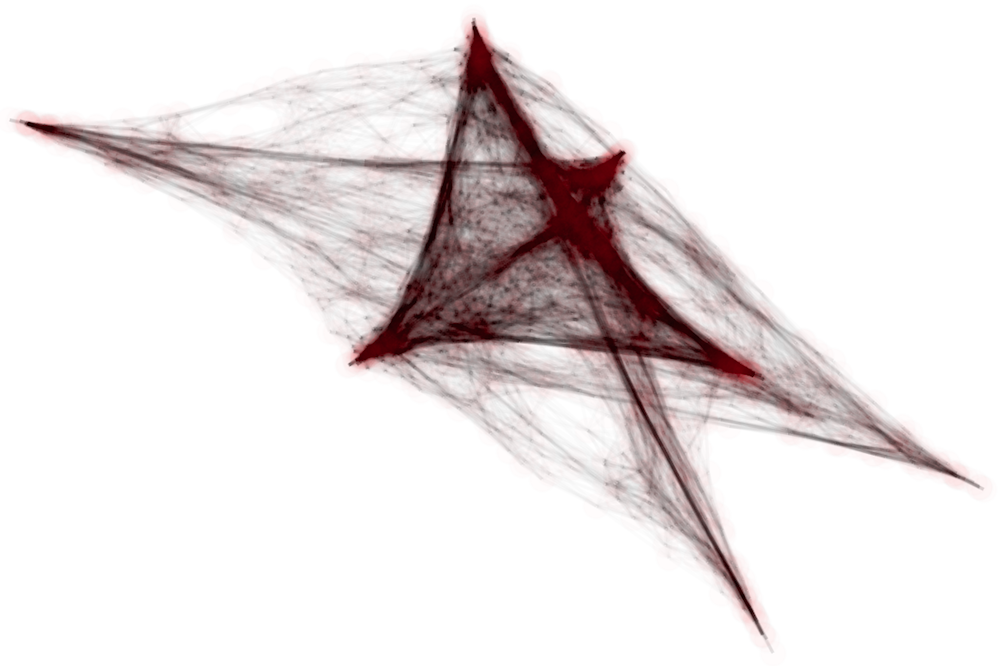}
    \end{subfigure}
    ~
    \begin{subfigure}{0.5\linewidth}
    \centering
    \includegraphics[height=2.5in]{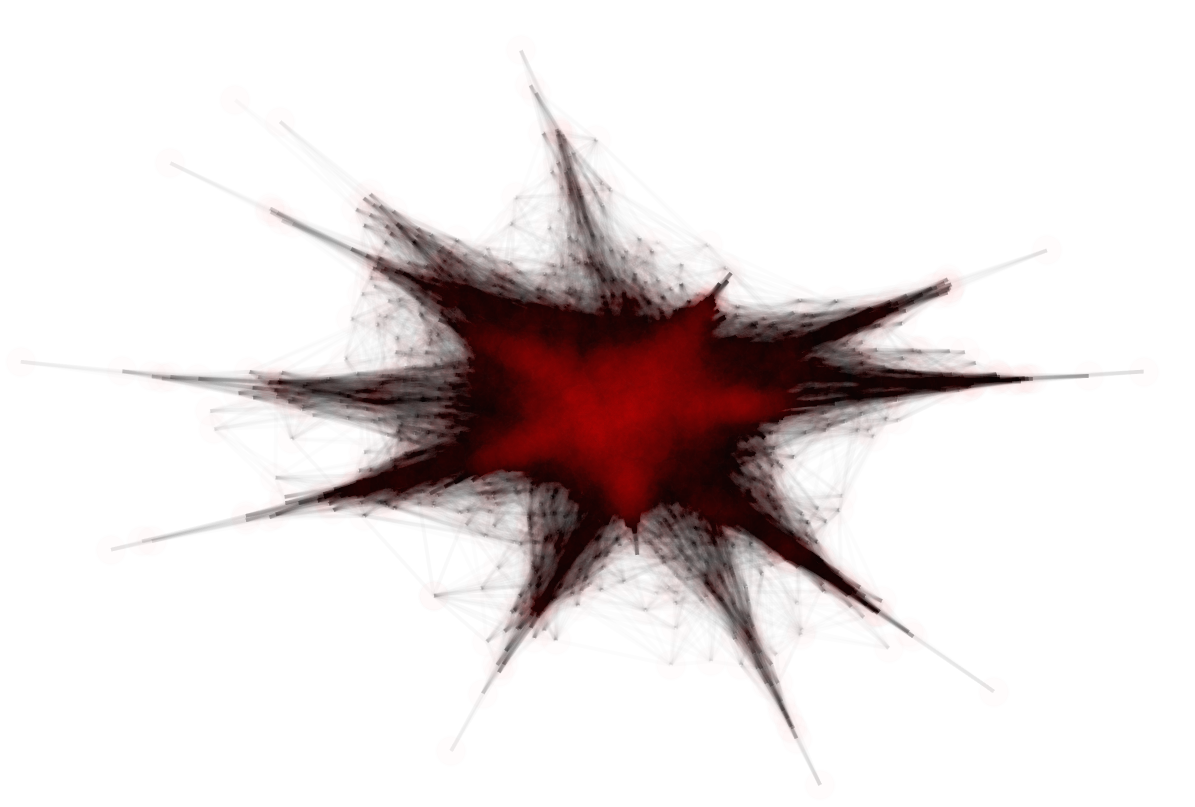}
    \end{subfigure}

    \caption{{\bf Discrete approximation of $\mathbb{P}\Sigma_m$, projective BHV space.} Unit norm phylogenetic trees on $m$ leaves are densely sampled and used as vertices in a graph weighted by pairwise BHV distances. All edges above a threshold distance are removed, with the threshold determined as the smallest value that maintains a single connected component. We call this graph the  $\epsilon$-net. The projective distance is then defined as the graph distance along the $\epsilon$-net, and we visualize a force directed layout of the graph. On the left we represent $\mathbb{P}\Sigma_4$, and on the right $\mathbb{P}\Sigma_5$. The case of four leaves is more easily appreciated as a join space, as described earlier.}
     \label{fig:epsilon_net}
\end{figure}

\subsection{The size of $\mathbb{P}\Sigma_m$}\label{sec:inj}

We now describe the size of $\mathbb{P}\Sigma_m$ (see also~\cite[3.3]{billera2001geometry} for a related discussion).
For simplicity, we will focus on the link in $\BHV_n$, which we will denote by $\aP_n$, and temporarily ignore the join with the simplex coming from the external edge lengths.
Observe that adjacent top-dimensional simplices in $\aP_n$ differ by a {\em rotation} of tree topologies, where a rotation collapses an internal edge and then expands out from the resulting node.
Next, recall that the homotopy type of $\aP_n$ is a wedge of $(n-1)!$ spheres of dimension $(n-3)$~\cite{robinson1996tree} (and see also~\cite[Thm. 6]{devadoss2014polyhedral}).
Moreover, we can explicitly describe these spheres, as follows.

As discussed in~\cite[Prop. 1]{devadoss2014polyhedral} and~\cite[\S 3.1]{billera2001geometry}, the boundary of the dual polytope to standard associahedron on $n$ letters (parametrizing parenthesizations of $n$ terms) embeds in many different ways into $\aP_n$.
Following~\cite{devadoss2014polyhedral}, let us denote this boundary by $\aK_n$.
Explicitly $\aK_n$ is a simplicial sphere of dimension $(n-3)$ where a $k$-simplex corresponds to a planar rooted tree with $n$ leaves and $k+1$ internal edges.
Then the homotopy type of $\aP_n$ can be described in terms of various embedded copies of $\aK_n$.
As a consequence, to understand the size of $\aP_n$, we need to compute the diameter of $\aK_n$.

For convenience, we describe this diameter in terms of counts of simplices; the actual value can then be obtained by multiplying by the diameter of a simplex.
In this guise, the problem is an old one which can be described in many different forms, perhaps most relevantly as the computation of maximal rotation distances between binary trees.

The main result here is that, for unrooted trees on $n$ leaves, the diameter of $\aK_n$ is $2n - 8$ for $n > 11$~\cite{pournin2014diameter}; this bound was established asymptotically (for sufficiently large but indeterminate $n$) in~\cite{sleator1988rotation}.
For smaller values, we have the following table (taken from~\cite[\S2.3]{sleator1988rotation}) of explicit values:

\begin{table}[ht]
    \caption{Diameters for $\aK_n$ for small values of $n$}
    \centering
    \begin{tabular}{c c c c c c c c}
    \hline\hline
    4 & 5 & 6 & 7 & 8 & 9 & 10 & 11 \\
    \hline
    2 & 4 & 5 & 7 & 9 & 11 & 12 & 15 \\
    \hline
    \end{tabular}
\end{table}

\begin{warning}
The results given in~\cite{pournin2014diameter} and~\cite{sleator1988rotation} differ slightly from the formula above and from each other due to divergent choices of indexing convention.
\end{warning}

More generally, the maximum rotation distance between labelled trees on $n$ leaves is $O(n \log n)$~\cite{sleator1992short}.
Of course, the cone point associated to the join with standard simplex means that the size is considerably smaller.


\section{Machine learning and statistical inference in $\Sigma_m$ and $\mathbb{P}\Sigma_m$}\label{sec:ML}

Our motivation for using the metric geometry of $\Sigma_m$ and $\mathbb{P}\Sigma_m$ comes from the problems of describing and comparing collections of trees generated from experimental data.
Regarding such collections as samples from distributions on the evolutionary moduli spaces, we are interested in basic statistical inference --- estimating parameters describing these distributions and determining if two samples came from the same or different distributions.
More generally, we would like to understand the kinds of distributions that can arise in evolutionary moduli spaces.
We are also interested in clustering and classification (i.e., unsupervised and supervised learning) problems in this context.
Given a set of unlabeled samples, we want to infer clusters of points that have similar clinical outcomes.
Given a set of labeled samples, we want to produce classifiers that can assign labels to new points in order to predict clinical outcomes.
In this section, we will review available tools for these kinds of problems.

\subsection{Statistics for distributions in $\Sigma_m$ and $\mathbb{P}\Sigma_m$}

In order to study probability distributions in evolutionary moduli spaces, it is necessary to have reasonable notions of moments of the distribution, expectation of random variables, and analogues of the law of large numbers.
Since $\Sigma_m$ and $\mathbb{P}\Sigma_m$ are $\CAT(0)$ spaces, points are connected by unique geodesics and there is a sensible notion of a centroid of a collection of points.
Discussion of statistical inference in $\Sigma_m$ was initiated in~\cite{billera2001geometry}, and subsequently Holmes has written extensively on this topic~\cite{holmes2003bootstrapping, holmes2003bootstrapping, holmes2005statistical} (and see also~\cite{feragen2013tree}).
More generally, Sturm explains how to study probability measures on general $\CAT(0)$ spaces~\cite{sturm2003probability}.
He shows that there are reasonable notions of moments of distribution, expectation of random variables, and analogues of the law of large numbers on $\CAT(0)$ spaces.

\begin{definition}
Given a fixed set of $n$ trees $\{T_0, \ldots T_{n-1}\} \subseteq \Sigma_m$, the Fr\'echet mean $T$ is the unique tree that minimizes the quantity \[E = \sum_{i = 0}^{n-1} d_{\Sigma_m}(T_i, T)^2.\] 
The variance of $T$ is the ratio $\frac{E}{n}$.
\end{definition}

Sturm provides an iterative procedure for computing the mean and variance of a set of points in $\Sigma_m$, and by exploiting the local geometric structure of $\Sigma_m$, Miller, Owen, and Provan produce somewhat more efficient algorithms for computing the mean~\cite{miller2012polyhedral}.
Furthermore, Sturm proves versions of Jensen's inequality and the law of large numbers in this context.
The situation for the central limit theorem is less satisfactory.
Barden, Le, and Owen study central limit theorems for Fr\'echet means in $\Sigma_m$~\cite{barden2013mean}; as they explain, the situation exhibits non-classical behavior and the limiting distributions depend on the codimension of the simplex in which the mean lies.
Finally, there has been some work on principal components analysis (PCA) in $\Sigma_m$~\cite{Nye2011PCA}.

However, in contrast to classical statistics on $\mathbb{R}^n$, we do not know many sensible analytically-defined distributions on the evolutionary moduli spaces.
Billera-Vogtmann-Holmes briefly introduce a family of Mallows distribution on $\Sigma_m$ with density function
\[
x(t) = \kappa e^{\alpha d_{\Sigma_m}(t_1, t)}
\] 
for fixed $t_1 \in \Sigma_m$, and an analogous family can be defined on $\mathbb{P}\Sigma_m$.
Sampling from these distributions is not easy; in general, the behavior of distributions on $\Sigma_m$ and sampling algorithms is somewhat perverse due to the pathological behavior near the origin due to the exponential growth in the mass of an $\epsilon$ ball.
A much more tractable source of distributions on $\Sigma_m$ and $\mathbb{P}\Sigma_m$ arise from resampling from a given set of empirical data points.

\subsection{Distributions in $\Sigma_m$ and $\mathbb{P}\Sigma_m$ via distributions in $\mathbb{R}^n$}

One way to grapple with the difficulties in dealing with distributions on $\Sigma_m$ and $\mathbb{P}\Sigma_m$ is to instead study associated projections into distributions on Euclidean space.
The advantage of this approach is evident; we are now in a setting where the theory of moments, the central limit theorem, and asymptotic consistency for resampling procedures are all very familiar.
Of course, it is important to keep in mind that the moments derived in this setting will reflect the geometry of the evolutionary moduli space is complicated ways, and inverses to the projections will not usually exist.
Nonetheless, for purposes of many kinds of statistical tests (e.g., hypothesis testing about distributions generating observed samples), this approach can be very effective.
There are a number of natural ways to map metric measure spaces into Euclidean space; in this section, we discuss several strategies derived from the use of the metric.

An intrinsic map comes from looking at the distance distribution on $\mathbb{R}$ induced by $\partial_M$.
Specifically, given a Borel distribution $\Psi$ on $(M, \partial_M)$, the product distribution $\Psi \otimes \Psi$ on $M \times M$ induces a distibution on $\mathbb{R}$ via $\partial_M$.
Applying this construction to the empirical measure on a finite sample yields the empirical distance distribution.
More generally, for any fixed $n$, we can consider the distribution on $\mathbb{R}^{n^2}$ induced by taking the product distribution $\Psi^{\otimes n}$ on $M^{\times n}$ and applying $\partial_M$ to produce the $n \times n$ matrix of distances.
Gromov's ``mm-reconstruction theorem'' showed that in the limit as $n \to \infty$, the distance matrix distributions completely characterize the distribution $\Psi$ on $(M, \partial_M)$~\cite{gromov1981}.
Once again, given a sufficiently large finite sample, we can construct the empirical distance matrix distributions for any fixed $n$.

Another approach involves choosing a fixed set of $n$ landmarks and considering the vector of distances from a fixed point to the landmarks.
Given a set $L = \{t_1, t_2, \ldots, t_n\} \subset M$, there is a continuous map
\[
d_L \colon M \to \mathbb{R}^n
\]
specified by
\[
x \mapsto (\partial_M(x,t_1), \partial_M(x,t_2), \ldots, \partial_M(x,t_n)).
\]
Pushforward along $d_L$ again induces a distribution on $\mathbb{R}^n$ from one on $M$.
One expects that as $k$ increases (provided the landmarks are ``generic''), the induced distributions in $\mathbb{R}^n$ will characterize the distribution on $M$.

In both cases, choice of the parameter $k$ depends on some sense of the intrinsic dimension of the support of the distribution as well as the number of points available (in the case of finite samples).
Unfortunately, the required $k$ may well be quite large.

Finally, there is a substantial body of work on low-distortion embeddings of finite metric spaces into $\ell^p$ spaces, in particular Euclidean spaces.
Recall that the distortion of a non-contractive (distance expanding) embedding of metric spaces $f \colon (X, \partial_X) \to (Y, \partial_Y)$ is given by $\sup_{x_1 \neq x_2} \frac{\partial_Y (f(x_1), f(x_2))}{\partial_X (x_1, x_2)}$.
Notably, Abraham, Bartal, and Neiman show that one can construct a probabilistic embedding of an $n$-point finite metric space into an $O(\log n)$ dimensional space with distortion $O(\log n)$.
Pushing forward distributions on $\Sigma_m$ and $\mathbb{P}\Sigma_m$ provide another way of reducing statistical questions to Euclidean space.

\subsection{Distinguishing samples from different underlying distributions}

Given a set of samples $X \subset \Sigma_m$ and a partition $X = \aC_1 \cup \aC_2 \ldots \cup \aC_n$ (where $\aC_i \cap \aC_j = \emptyset$), it is often useful to be able to determine whether or not the different $\aC_i$ were generated from the same or different underlying distributions.
For instance, $\aC_1$ might represents samples from patients who received treatment and $\aC_2$ is untreated patients, or the different groups $\aC_i$ represent different observed genetic markers.
Based on the discussion of the previous two subsections, we can study this problem directly in $\Sigma_m$ or in via projections to $\mathbb{R}^n$.

Te Fr\'echet mean and variance provides a summary of each collection of samples $\aC_i$.
A standard comparison between groups is then given by the distance 
\[
\theta_{ij} = d_{\Sigma_m}(T(\aC_i), T(\aC_j))
\]
between the means.
In order to understand the variability due to sampling, we can use bootstrap resampling (or more general $k$ out of $n$ resampling without replacement) to generate confidence intervals for the value of $\theta_i$.
Asymptotic consistency for the bootstrap follows from the fact that the VC dimension of the collections of balls in $\Sigma_m$ and $\mathbb{P}\Sigma_m$ is bounded, via the usual criteria~\cite{Gine1984, Gine1986, Gine1990}.

However, it is often simpler to consider tests induced by the projections into $\mathbb{R}^n$ discussed above.
Here, we can compare collections $\aC_i$ and $\aC_j$ by using any of the many standard non-parametric comparison techniques for real distributions, for example $\chi^2$ tests or two-sample Kolmogorov-Smirnov tests.

One pervasive problem in clinical applications is that often the number of samples is quite small, and so we are often far from the asymptotic regime for statistical tests.
Standard small-sample corrections can be applied.
However, for this reason a machine learning approach to analyzing the data is often more useful.

\subsection{Clustering in $\Sigma_m$ and $\mathbb{P}\Sigma_m$}

Given the difficulties with statistical inference in $\Sigma_m$ and $\mathbb{P}\Sigma_m$, it is useful to complement these approaches with techniques from machine learning.
The most basic family of techniques we might consider is clustering, a kind of unsupervised learning.
Here, given a finite set $X$ in $\Sigma_m$ or $\mathbb{P}\Sigma_m$, we search for a partitionings of the points into clusters which optimize some criterion for the ``goodness'' of the clustering.

Regarding $\Sigma_m$ and $\mathbb{P}\Sigma_m$ simply as metric spaces, we can apply standard clustering algorithms that operate on arbitrary metric spaces.
For example, we can apply standard $k$-means clustering, using the centroids as defined above.
A related alternative is the $k$-medoids algorithm.
Like $k$-means, $k$-medoids seeks partitions which are optimal in the sense of minimizing the sum of squared distances; the cost function for a cluster $C = \{x_i, x_j\}$ is given by $\sum_{i < j} \partial(x_i,x_j)^2$.
But instead of using cluster centroids as in $k$-means, cluster assignments are determined by medoids, which are points $z \in C$ that minimize $\sum_i d(z,i)$.
The advantage of $k$-medoids over $k$-means is that the problem of finding a centroid in $\Sigma_m$ or $\mathbb{P}\Sigma_m$ is avoided.

Another natural family of clustering algorithms comes from spectral clustering techniques.
Recall that spectral clustering can be applied to finite subsets of any metric space; one constructs an embedding into Euclidean space using the graph Laplacian associated to a graph encoding the local metric structure of the set of points and then performs $k$-means clustering.
As such, spectral clustering can be applied both to $\Sigma_m$ and $\mathbb{P}\Sigma_m$.
However, as illustrated in Figure~\ref{fig:embedding}, there is substantial distortion under such embeddings for low dimensions.

\begin{figure}
    \centering
    \includegraphics[width=\linewidth]{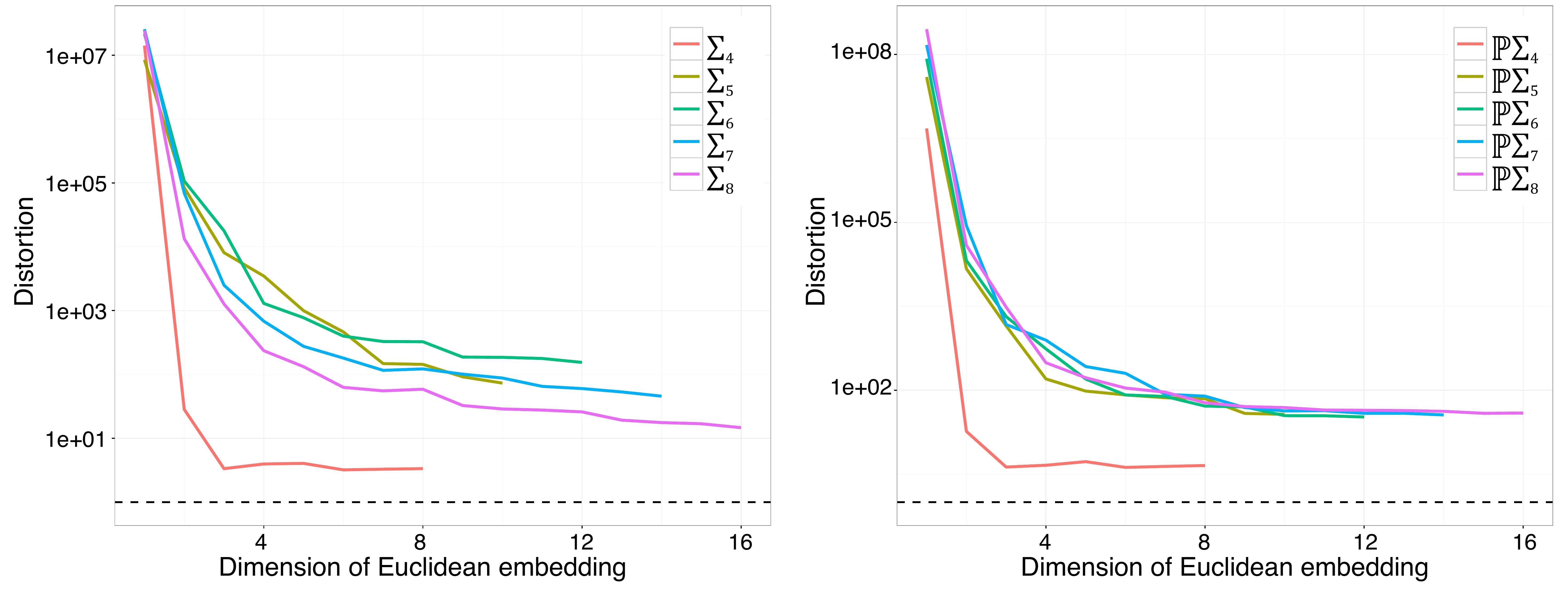}
    \caption{{\bf Euclidean embedding of the affine and projective tree spaces respectively, $\Sigma_m$ and $\mathbb{P}\Sigma_m$.} A spectral embedding approach is taken to finding Euclidean approximations of our $\epsilon$-nets. We are interested in the relation between $m$ and the smallest $d$ with acceptable distortion of embedding into $\mathbb{R}^d$.}
    \label{fig:embedding}
\end{figure}

\subsection{Supervised Learning}

Although clustering algorithms are very useful for exploratory data analysis, for clinical applications we expect that classification problems are more salient.
Specifically, a temporal sequence of tumor samples will be linked with a categorical or numeric label denoting the clinical management of the patient.
We would then like to predict patient outcomes or expect response to treatment using a discriminative supervised learning algorithm operating in $\Sigma_m$ or $\mathbb{P}\Sigma_m$.
Analogous to our use of $k$--medoids clustering for unsupervised grouping, the most basic algorithm for supervised learning is a $k$--nearest neighbor ($k$--NN) predictor.
In this algorithm the predicted label for a given point is generated by taking a majority or weighted vote over the labels of the $k$ nearest trees.
The optimal value of $k$ then specifies an order--$k$ Voronoi tesselation of the space that provides a description of the sizes of the predictive neighborhoods surrounding each element of the data set.
We use this classification algorithm to study clinical correlates of trees determined by tumor samples from glioma patients; see Figure~\ref{fig:gliomaTMZ}.


\section{Tree dimensionality reduction}\label{sec:treedimred}

When analyzing a large number of genomes, phylogenetic trees are often too complex to visualize and analyze as they can contain thousands of branches.
In this section, we will explain a technique for dimensionality reduction that projects a single tree in $\Sigma_m$ or $\mathbb{P}\Sigma_m$ to a ``forest'' of trees in $\Sigma_{n}$, for $n < m$.
The main idea is that by subsampling leaves of a large tree we can have a distribution of smaller trees that can capture properties of the more complex structure.
This procedure makes it easy to visualize and analyze high-dimensional data, and avoids scalability issues with algorithms for working with the spaces of phylogenetic trees.
We believe that the analysis and visualization of the resulting clouds of trees is an effective way to study high-dimensional evolutionary moduli spaces.
To provide theoretical justification for this claim, we prove that this procedure is stable, in the sense that it preserves distances up to a constant factor.

\subsection{Structured dimensionality reduction}

Let $\aE_m$ denote either $\Sigma_m$ or $\mathbb{P}\Sigma_m$.

\begin{definition}
For $S \subseteq \{1,\ldots,m\}$, define the tree projection function 
\[
\Psi_S \colon \aE_m \to \Sigma_{|S|}
\] 
by specifying $\Psi_S(T)$ to be the unique tree obtained by taking the full subgraph of $t$ on the leaves that have labels in $S$ and then deleting vertices of degree $2$.
An edge $e$ created by vertex deletion is assigned weight $w_1 + w_2$, where the $w_i$ are the weights of the incident edges for the deleted vertex.
(It is easy to check that the order of vertex deletion does not change the resulting tree.)
\end{definition}

A representative example of $\Psi_S$ is shown in Figure~\ref{fig:dimred_tree-dim-red}.
\begin{figure}
    \centering
    \includegraphics[width=6in]{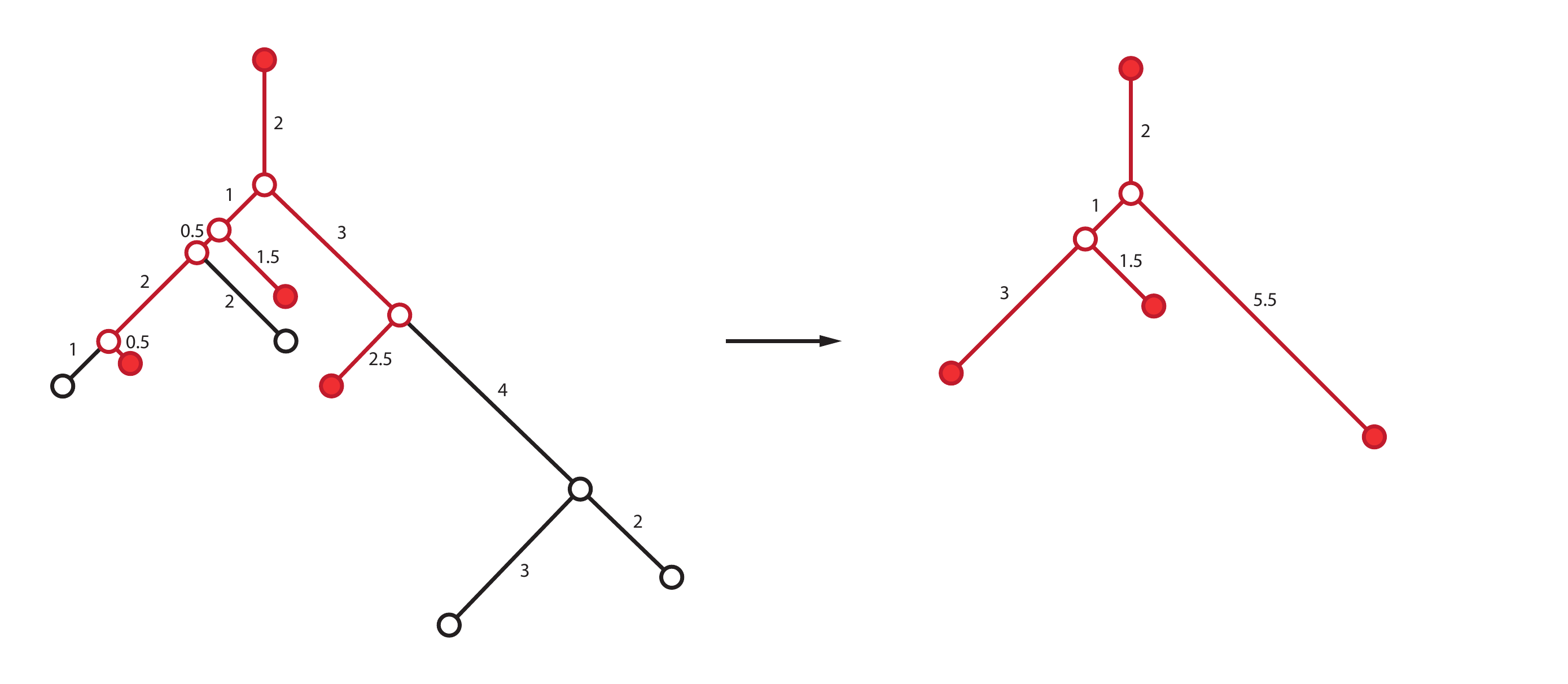}
    \caption{{\bf Tree dimensionality reduction.} The leaves that are not highlighted in the starting phylogeny are pruned and their external edges removed; internal vertices of degree $2$ are collapsed and edge weights on either side are summed.}
    \label{fig:dimred_tree-dim-red}
\end{figure} 

Using $\Psi$, we can describe a number of dimensionality reduction procedures.
The most basic example is simply to exhaustively subsample the labels.
Let $\aD(\Sigma_m)$ denote the set of distributions on $\Sigma_m$.

\begin{definition}[Tree dimensionality reduction]
For $1 \leq k < m$, define the map
\[
\Psi_k \colon \aE_m \to \aD(\Sigma_k)
\]
as the assignment that takes $T \in \aE_m$ to the empirical distribution induced by $\Psi_S$ as $S$ varies over all subsets of $\{1,\ldots,m\}$ of size $k$.
Define the map
\[
\Psi'_k \colon \aE_m \to \prod_{S \subseteq \{1,\ldots,m\}, |S| = k} \Sigma_k 
\]
as the map that takes $T \in \aE_m$ to the product of $\Psi_S(T)$ as $S$ varies over all subsets of $\{1, \ldots, m\}$ of size $k$.
\end{definition}

(In practice, we approximate $\Psi_k$ using Monte Carlo approximations.)

Often there is additional structure in the labels that can be exploited.
For instance, in many natural examples, the genomic data has a natural chronological ordering.
When this holds, sliding windows over the labels induces an ordering on subtrees generated by $\Psi_S$.
Rather than just regarding such a sequence as a distribution, the ordering makes it sensible to consider the associated trees as forming a piecewise-linear curve in $\Sigma_k$.
(Note that given a set of points in $\Sigma_k$ it is always reasonable to form the associated piecewise-linear curve because each pair of points is connected by a unique geodesic.)
Let $\aC_k(\Sigma_m)$ denote the set of piecewise linear curves in $\Sigma_m$; equivalently, $\aC_k(\Sigma_m)$ can be thought as the set of ordered sequences in $\Sigma_m$ of cardinality $k$.
A schematic example of this sequential operation is given in Figure~\ref{fig:dimred_def4}.

\begin{figure}
    \centering
    \includegraphics[width=4in]{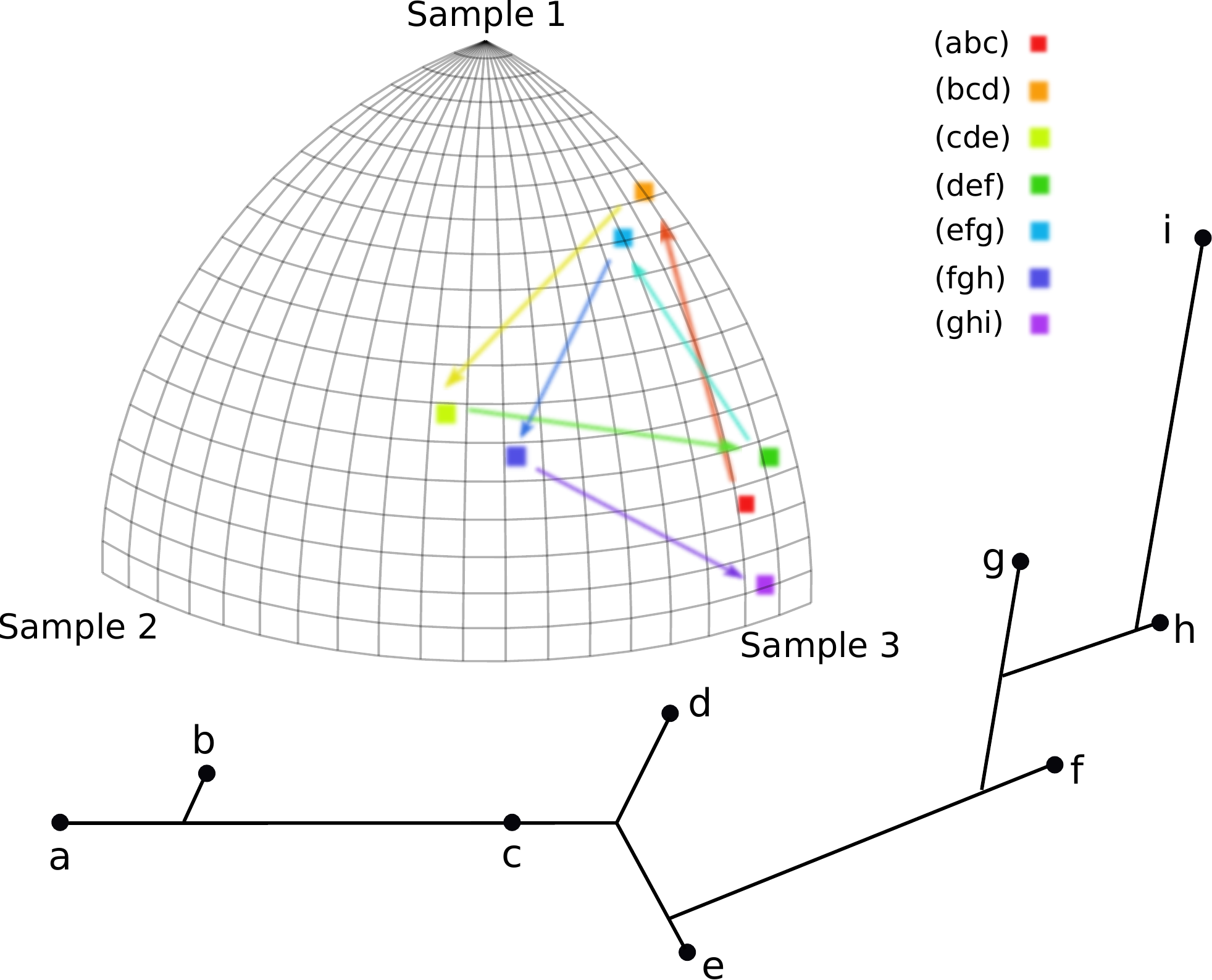}
    \caption{{\bf Sequential tree decomposition on a set of ordered samples.} The subsets of trees generated from the initial phylogeny respect the ordering on the leaves, as would be the case in a temporally ordered set of samples. Each subtree can be visualized on a common set of axes, to chart motion through time. Periodicity in the sequence of branch lengths, for example, might give rise to cycles in the evolutionary moduli space.}
    \label{fig:dimred_def4}
\end{figure} 

\begin{definition}[Sequential tree dimensionality reduction]
For $1 \leq k < m$, define the map
\[
\Psi_C \colon \aE_m \to \aC_{m-k}(\Sigma_k)
\]
as the assignment that takes $T \in \aE_m$ to the curve induced by $\Psi_S$ as $S$ varies over the subsets $\{1, \ldots, k\}$, $\{2, \ldots, k+1\}$, etc.

Equivalently, we can regard this as producing a map
\[
\Psi_C \colon \aE_m \to \aD(\Sigma_k).
\]
\end{definition}

There are many variants of $\Psi_C$ depending on the precise strategy for windowing that is employed.

\subsection{Tree dimensionality reduction and neighbor-joining}

We have described tree dimensionality reduction in terms of the map $\Psi$, which is an operation on tree spaces.
In practice, this technique would be applied by producing a very large phylogenetic tree from the raw data and subsequently applying $\Psi$.
However, there is an alternative form of tree dimensionality reduction that instead subsamples the raw data to produce smaller phylogenetic trees.
In this section, we discuss the relationship between these two procedures in the context of neighbor-joining.

Neighbor-joining is an algorithm for producing a tree from metric data (or more broadly, a dissimilarity measure)~\cite{felsenstein2003}.
That is, the input is a set of points $X$ and a metric $\partial_X \colon X \times X \to \mathbb{R}$.
One of the main theorems about consistency of neighbor-joining is that when $\partial_X$ is ``close'' to a tree metric $\partial_T$, neighbor-joining recovers $T$.
Recall that given a tree $T$, the associated metric $\partial_T$ is defined by taking the distance between leaves $i$ and $j$ to be 
\[
\partial_T(i,j) = \sum_{e \in P_{ij}} \ell(e),
\]
where $P_{ij}$ is the unique path in $T$ from $i$ to $j$ and $\ell(e)$ is the length of the edge $e$.

The specific consistency theorem we use is due to Atteson~\cite{atteson}: if
\[
\max_{x_i,x_j \in X} |\partial_X(x_i,x_j) - \partial_T(x_i,x_j)| \leq \frac{1}{2} \min_{e \in T} \ell(e),
\]
then neighbor-joining recovers $T$.
In this case we say that $(X, \partial_X)$ is consistent with $T$.

\begin{proposition}\label{prop:subnj}
If $(X,\partial_X)$ is consistent with $T$, for any subset $S = \{x_1, x_2, \ldots x_k\} \subseteq X$, the associated submetric space is consistent with $\Psi_S(T)$.
\end{proposition}

\begin{proof}
It is clear from the definition of $\Psi_S(T)$ that 
\[
\min_{e \in \Psi_S(T)} \ell(e) \geq \min_{e \in T} \ell(e).
\]
On the other hand, 
\[
\max_{x_i,x_j \in S} |\partial_X(x_i,x_j) - \partial_T(x_i,x_j)| \leq
\max_{x_i,x_j \in X} |\partial_X(x_i,x_j) - \partial_T(x_i,x_j)| 
\]
The result follows.
\end{proof}

For a metric space $(X,\partial_X)$, let $T(X)$ denote the tree obtained from neighbor-joining applied to $(X,\partial_X)$.
As a consequence of Proposition~\ref{prop:subnj}, given a metric space $(X, \partial_X)$ that is consistent with $T(X)$, the distribution $\Psi_k(T(X))$ is identical to the distribution $\{T(X_S)\}$ where $S$ varies over all subsets of $X$ of cardinality $k$ and $X_S$ denotes the metric space structure on $S$ induced by $\partial_X$.

\begin{remark}
One can ask the same question for other methods of producing phylogenetic trees; the situation is substantially more complicated, and we intend to provide a detailed analysis in future work.
\end{remark}

\subsection{Stability of tree dimensionality reduction}

In order to apply tree dimensionality reduction in the face of potentially noisy data, we would like to know that small random perturbation of the original sample results in a distribution of subsamples that is ``close'' in some sense (e.g., small shifts in the centroid in $\Sigma_m$).
Conversely, if two distributions of subsamples are suitably close, we would like to be able to conclude that the sampled trees are also close.

We begin with a lemma describing the interaction of $\Psi_S$ and the boundaries of orthants.

\begin{lemma}\label{lem:rotprojcom}
Fix $S \subseteq \{1,\ldots,m\}$.
Let $T$ be a point in the interior of an orthant of $\Sigma_m$, and let $\gamma \colon [0,1] \to \Sigma_m$ be the geodesic path contained in that orthant from $T$ to $T'$, where $T'$ is obtained from $T$ by collapsing an interior edge to length $0$.
Then $\Psi_S(T')$ is obtained from $\Psi_S(T)$ by shrinking an interior edge, and $\Psi_S(\gamma)$ is the geodesic path from $\Psi_S(T)$ to $\Psi_S(T')$.
\end{lemma}

\begin{proof}
It suffices to show that $\Psi_{S}(T')$ is obtained from $\Psi_S(T)$ by shrinking an edge; given this, the assertion about $\gamma$ is clear.
Let $e = (v_1,v_2)$ denote the edge to be collapsed, with $v_1$ the vertex closer to the root and $v_2$ the vertex closer to the leaves.
There are three possibilities.
If the edge $e$ is not present in $\Psi_S(T)$, then this means that none of the leaves below $e$ are in $S$; as a consequence, none of the leaves below $v_1$ in $T'$ are in $S$, and so $\Psi_S(T')$ will also not contain $e$ and so $\Psi_S(T) = \Psi_S(T')$.
If the edge $e$ is present in $\Psi_S(T)$ and does not participate in a vertex collapse, this means that both edges emanating from $v_2$ are present in $\Psi_S(T)$ and therefore that collapsing $e$ to $0$ commutes with applying $\Psi_S$.
Finally, if applying $\Psi_S$ to $T$ causes $e$ to be concatenated with another edge, then there are two cases to analyze --- $e$ could be concatenated via the deletion of $v_1$ or $v_2$.
Suppose that the concatenation occurs because the other ``downward'' edge with endpoint $v_2$, which we will denote $e'$, leads to leaves that are not in $S$.
If we collapse $e$ to $0$ before applying $\Psi_S$, $e'$ will still be deleted when we apply $\Psi_S$, and so the result will be the same.
The case of deletion of $v_1$ is analogous.
\end{proof}

In light of Lemma~\ref{lem:rotprojcom}, the projection $\Psi_S$ preserves paths.
The other thing we need to understand is the potential increase in length caused by applying $\Psi_S$.
Specifically, we need to consider the impact of the addition of edge lengths that occurs when a degree 2 vertex is produced by the reduction process.
In the simplest case, we are considering the map $\mathbb{R}^2 \to \mathbb{R}$ specified by $(x_1,x_2) \mapsto x_1+x_2$, and in general, we are looking at $\mathbb{R}^n \mapsto \mathbb{R}$ specified by $(x_1, x_2, \ldots, x_n) \mapsto \sum_{i=1}^n x_i$.
Squaring both sides, it is clear that
\[
\partial_{\mathbb{R}^n}((x_i), (y_i))^2 \leq \partial_{\mathbb{R}}(\sum_{i=1}^n x_i, \sum_{i=1}^n y_i)^2.
\]
On the other hand, since 
\[
\partial_{\mathbb{R}}(\sum_{i=1}^n x_i, \sum_{i=1}^n y_i) \leq n(\max_i |x_i - y_i|) \leq n \partial_{\mathbb{R}^n}((x_i), (y_i)),
\]
the addition of edge lengths can result in an expansion bounded by the size of the sum.

\begin{remark}
Another way to interpret the previous result is to observe that the addition map is an isometry for the Manhattan distance (when working in the positive orthant) but not for the Euclidean distance.
\end{remark}

For a rooted tree $T$, let $\depth(T)$ denote the length of the longest path from a leaf to the root.

\begin{proposition}\label{prop:projcont}
Let $S \subseteq \{1,\ldots,m\}$ such that $|S| > 1$ and let $\gamma \colon [0,1] \to \aE_m$ be a path from $T$ to $T'$.
Then $\gamma \circ \Psi_S \colon [0,1] \to \Sigma_{|S|}$ is a path from $\Psi_S(T)$ to $\Psi_S(T')$ and $|\gamma'| \leq \max(\depth(T), \depth(T')) |\gamma|$.
\end{proposition}

\begin{proof}
First, observe that if $T$ and $T'$ are in the same orthant of $\aE_m$, the result is clear.
In this case, for any $S$, $\Psi_S(T)$ and $\Psi_S(T')$ will be in the same orthant of $\Sigma_{|S|}$.
By the discussion above, the length of the projected path in that orthant is bounded by the length of the path in $\aE_m$ scaled by the depth of the tree.
This argument also shows that result holds for trees joined by the cone path; $\Psi_S$ applied to the cone point produces the cone point.

Now suppose that $T$ and $T'$ are not in the same orthant and neither $T$ nor $T'$ is contained in a positive codimension subspace of $\aE_m$ (i.e., they are not on the boundary of any orthant).
Further, we assume that $\gamma$ does not go through the origin and that $\gamma$ can be expressed in terms of a sequence of contractions and expansions of a single edge.
That is, we assume that $\gamma$ only goes through codimension $1$ faces of each orthant.
It suffices to consider this case, since for a general $\gamma$ that does not go through the origin but might traverse faces of codimension larger than $1$, observe that for any $\epsilon > 0$, we can perturb $\gamma$ to produce a path $\gamma'$ with the same endpoints which satisfies the hypothesis above and has $|\gamma'| = |\gamma| + \epsilon$.
Passing to limits then implies that the bound holds for such a path.
Similarly, a limit argument implies the result for a path that starts or ends on a positive codimension subspace of $\aE_m$.
Moroever, given this case, more complicated paths that involve both rotations and also pass through the origin satisfy the bound by an easy induction.

Thus, fix a subset $S \in \{1,\ldots,m\}$.
Lemma~\ref{lem:rotprojcom} now implies that $\Psi_S(\gamma)$ is a path from $\Psi_S(T)$ to $\Psi_S(T')$, and the discussion preceding the proposition implies that the potential expansion in length is $\max(\depth(T),\depth(T'))$.
\end{proof}

\begin{remark}
In fact, the expansion factor in Proposition~\ref{prop:projcont} depends on the number of edge conactenations that occur when $\Psi_S$ is applied; in situations where an estimate of this is available, tighter bounds can be used.
\end{remark}

Using Proposition~\ref{prop:projcont}, it is straighforward to deduce the next two theorems that provide the theoretical support for the use of tree dimensionality reduction.
The following theorem is an immediate consequence of Proposition~\ref{prop:projcont}, choosing the path realizing the distance between $T$ and $T'$.

\begin{theorem}\label{thm:stab}
For $T,T' \in \aE_m$ such that $d_{\aE_m}(T,T') \leq \epsilon$, then for any $S \subseteq \{1,\ldots,m\}$ such that $|S| > 1$, 
\[
d_{\Sigma_{|S|}}(\Psi_S(T), \Psi_S(T')) \leq \max(\depth(T),\depth(T')) \epsilon.
\]
Moreover, this bound is tight.
\end{theorem}

Let $A$ and $B$ be subsets of $\aE_m$ such that each item of $A$ and $B$ has a label in $L \subset \aP(\{1,\ldots,m\})$ (where $\aP(-)$ denotes the power set of $\{1,\ldots,m\}$).
Then we can define a matching distance as
\[
d_{M,L}(A,B) = \max_{S \in L} d_{\aE_m}(A(S), B(S)).
\]

Without assuming such a labelling, we define the matching distance between $A$ and $B$ to be 
\[
d_M(A,B) = \min_{\phi} \max_{a \in A} d_{\aE_m}(a,\phi(a)),
\]
where $\phi$ varies over all bijections $A \to B$.

The following is now also immediate from Proposition~\ref{prop:projcont}.

\begin{lemma}\label{lem:converse}
For $T, T' \in \aE_m$ and $L \subset \aP(\{1,\ldots,m\})$, , 
\[
d_{\aE_m}(T,T') \geq \left(\frac{1}{\max(\depth(T), \depth(T'))}\right) d_{M,L}(\Psi_{S \in L}(T), \Psi_{S \in L}(T')).
\]
\end{lemma}

\begin{figure}
    \begin{subfigure}{\linewidth}
    \centering
    \includegraphics[width=0.7\linewidth]{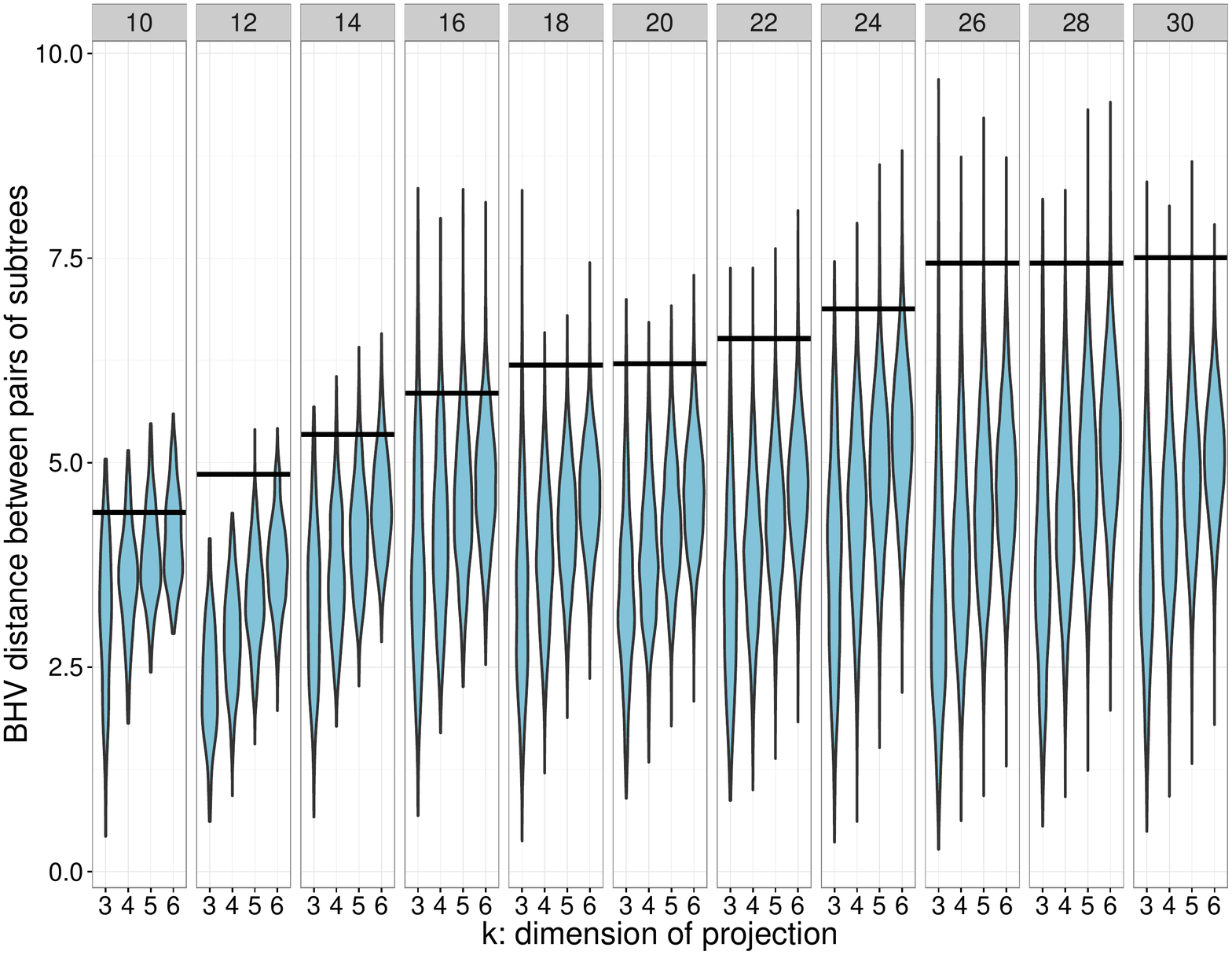}
    \end{subfigure}
    
    \begin{subfigure}{\linewidth}
    \centering
    \includegraphics[width=0.7\linewidth]{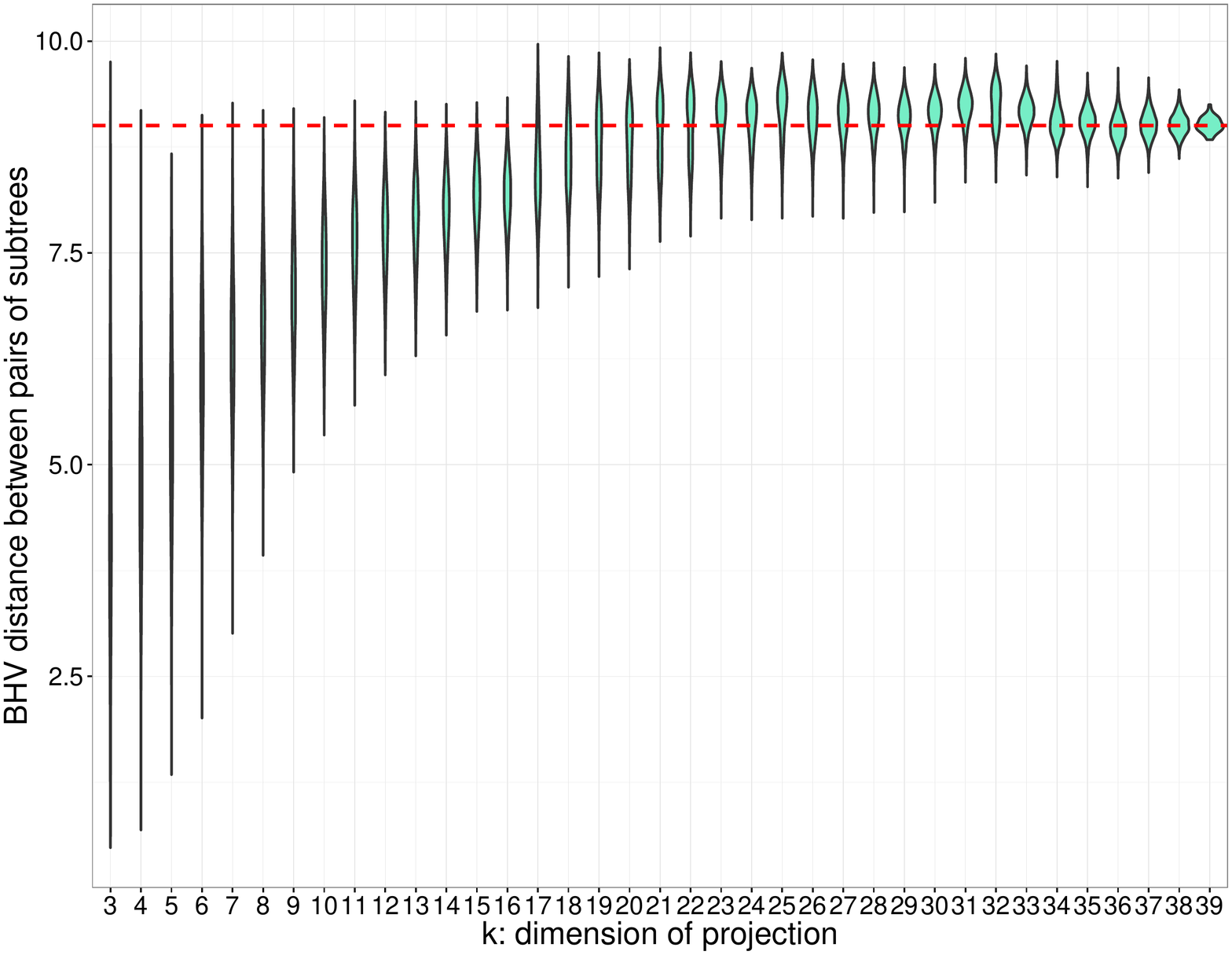}
    \end{subfigure}

    \caption{{\bf Distributions of subsample distances under the tree projection operation.} (Top) Pairs of $m$-dimensional phylogenies with known distance (horizontal black lines) are projected into distributions of low dimensional trees. The distances between elements of the projections can exceed the original inter-phylogeny distance ($\epsilon$), but rarely approach the upper bound. (Bottom) As the dimension of the projection operator approaches that of the initial phylogenies, there is a decrease in the variance of the distribution of subtree distances and its median approaches the 40-dimensional $d_{BHV}(T, T')$.}
    \label{fig:dimred_calibration}
\end{figure}

We conclude the discussion by describing some computational results on simulated data that illustrates the divergence between $d_M$ and $d_{\aE_m}(T,T')$; see Figure~\ref{fig:dimred_calibration}.
To demonstrate the stability of the tree projection operation, and to empirically test its approach of the distance bound, we constructed a panel of pairs of $m$-dimensional trees, $m > 10$.
The distances between the pairs of trees were computed and compared against the distributions of distances induced by the $k$-dimensional projection operator, $\Psi_k$ with $k \in {3,4,5,6}$.
The distances between elements of the projections can exceed the original inter-phylogeny distance ($\epsilon$), but rarely approach the upper bound.
To further characterize the behavior of the distribution of subtree distances as a function of $k$, we compared the distributions of projected distances ranging from $k=3$ to $k=m-1$ for a fixed pair of $m$-dimensional phylogenies.
We see that the median approaches the true $m$-dimensional $d_{BHV}(T, T')$ with larger values of $k$, while the variance appears to decrease monotonically.

\subsection{Using tree dimensionality reduction for inference and machine learning}

Broadly speaking, the various tree dimensionality reduction operators transform questions about comparison of trees or analysis of finite sets of trees to questions about comparisons and analysis of sets of clouds of trees.
This has several advantages.
First, when working with $\mathbb{P}\Sigma_m$, the resulting clouds live in $\Sigma_k$, and as discussed above analysis in the non-projectivized space can be simpler.
Second, when projecting to $\Sigma_3$ or $\Sigma_4$, both visualization and analysis is easier (particularly in $\Sigma_3$, since that has a Euclidean metric).
For example, in order to perform supervised classification on a labelled set of trees $X$, we can simultaneously solve classification problems in $\Psi'_k(X)$ and use majority voting in order to assign labels to new trees.
Another possibility is to perform hierarchical clustering on the clouds in $\Psi'_k(X)$, resulting in another tree, and use the distance in tree space as a test statistic to discriminate between clouds.
In Section~\ref{sec:flu} below we describe an application that involves producing a predictor for influenza vaccine effectiveness using the variance of the distribution produced by $\Psi$.

\begin{remark}
Another interesting direction of research is to consider the use of topological data analysis summaries (i.e., hierarchical clustering dendrograms or barcodes) for the clouds of projected points.
The stability results above easily imply stability results for the associated barcodes of the projections.
We intend to return to this subject in a subsequent paper.
\end{remark}

\section{Clinical utility of longitudinal tumor genomics}\label{sec:cancer}

Progression of cancer is believed to be intimately related to the accumulation of genomic alterations in tumor cells~\cite{nowell1976clonal}.
Mutations can spur proliferation, either via activation of an oncogene or inactivation of a tumor suppressor.
The spatial and temporal heterogeneity of tumors can be addressed by reconstructing the evolutionary history of tumors from different samples.
For each location and time point, one can define (partially or totally) the genotype of the dominant clone.
As the evolution proceeds in a clonal fashion, the relationships between dominant clones can be structured as a phylogenetic tree.
Questions about the nature of the evolutionary process, mechanisms of resistance, stratification of patient tumor histories, or prognosis can be formulated as a comparison between sets of trees.

For example, a first step towards personalizing cancer therapy is to monitor the mutational status of patients along the therapeutic course.
Genomic snapshots before and after administration of cytotoxic therapy can reveal the extent of population remodeling.
A further goal is to establish \textit{in vivo} mouse models of every patient's tumor, as a means of rapidly exploring drug susceptibility and resistance.
Such models can be created by direct implantation of human tumor tissue into immunodeficient mice and are termed patient-derived xenografts (PDX).

In this section, we describe three applications of the mathematical machinery of evolutionary moduli spaces.
We begin with the medically relevant problem of how therapy affects the evolution of common leukemias.
Then, using public data from relapsed glioma, we highlight clinical correlates to the spatial distribution of patients in the moduli space.
Finally, we study single cell data from breast cancer derived xenografts to observe the departure of the tumor genetics from the primary lesion.

\subsection{Evolution of chronic lymphocytic leukemia under therapy}\label{sec:CLL}

Chronic lymphocytic leukemia (CLL) is the most common leukemia in adults, primarily affecting the elderly population (median age at diagnosis is 70)~\cite{smith2011incidence}.
CLL is a proliferative disorder of B-lymphoctyes characterized by a steady accumulation of clonal, non-functional B-cells.
Treatment strategies vary greatly given the heterogeneity in disease course, ranging from watchful waiting, to localized radiation, to systemic chemotherapy.
The fact that CLL is a relatively indolent malignancy makes it an excellent model for studying clonal evolution under different therapeutic strategies~\cite{wang2015tumor}.

A recent genomic study~\cite{landau2013evolution} performed whole exome sequencing on 160 CLL cases covering the spectrum of clinical courses the disease can take.
This data established a space of recurrent alterations, which was then used to genotype 18 patients for whom two time points were available.
Of these 18 patients, 10 of 12 treated with chemotherapy underwent clonal evolution compared to only 1 of 6 receiving no treatment according to the authors of the study.
We combine the 18 patients of \cite{landau2013evolution} with those of a similar study~\cite{schuh2012monitoring} wherein 3 CLL patients received chemotherapy and were sequenced at multiple time points.
The multiple time points for the 3 patients studied in \cite{schuh2012monitoring} are decomposed into all combinatorial triplets.
Therefore, there are a total of 18 phylogenetic trees inferred from \cite{landau2013evolution} and 12 phylogenetic trees inferred from~\cite{schuh2012monitoring}.
In Figure~\ref{fig:chemoCLL}, we map this data to $\mathbb{P}\Sigma_3$ and color based on treatment status.

\begin{figure}
    \begin{subfigure}{0.5\linewidth}
    \centering
    \includegraphics[height=2.2in]{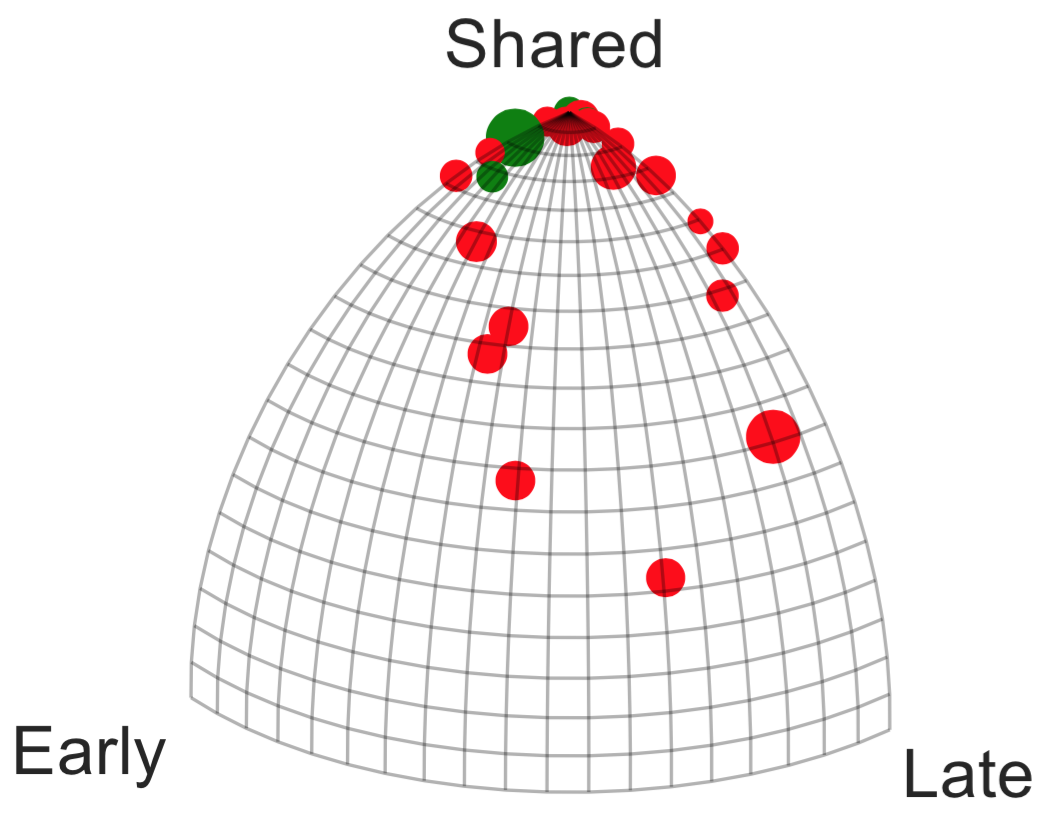}
    \end{subfigure}
    ~
    \begin{subfigure}{0.5\linewidth}
    \centering
    \includegraphics[height=2.2in]{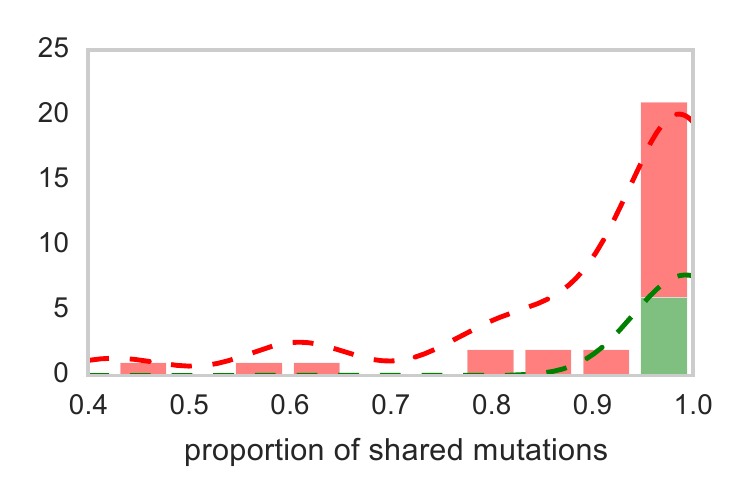}
    \end{subfigure}

    \caption{{\bf Different evolutionary patterns observed in CLL.} (Left) In patients that did not receive chemotherapy the tumor exome did not change. Tumors from these patients, represented in green, shared most mutations along different time points. However, under therapy (shown in red) mutations before therapy were not present after treatment, and new mutations were acquired. (Right) The proportion of mutations shared between time points is significantly lower among patients treated with chemotherapy ($p = 0.049$, log-rank test).}
    \label{fig:chemoCLL}
\end{figure}

From a clinical standpoint, a central question is whether treatment promotes evolution of the cancer and whether there is strong evidence for avoiding cytotoxic therapy in patient management.
Quite clearly the distribution of 6 untreated patients resembles a pattern (in green) forms a tight cluster, indicating that tumors from patients who did not received therapy are stable genetically, sharing most of the mutations.
However, in red are represented the histories of tumors of 15 patients under therapy, presenting some mutations at different times that are not shared across different samples.
The ratio between the number of mutations that are exclusive in the early branch versus the ones that are share with other phases is represented in right hand side of Figure~\ref{fig:chemoCLL}.
A number close to zero indicates a genetically stable tumor.

To assess how different are the clonal histories of tumors from untreated vs treated patients, we studied the distance between the centroids of the two populations regarded as points in $\mathbb{P}\Sigma_3$.
The 95\% CI for the distance between the centroids of the treated / untreated groups is (0.15, 0.36), under 1000-fold bootstrap resampling.
The analogous intervals for untreated / untreated and treated / treated are (0.01, 0.14) and (0.02, 0.16) respectively.
This analysis shows that the centroids of these clinically distinct sets of patients are well-resolved, supporting the idea that untreated tumors are more stable than treated ones, where district mutations can appear along the evolution of the tumor.

\subsection{Tree geometry associated with grade at relapse in gliomas}\label{sec:glioma}

As another application, we examined low grade gliomas (LGG), a set of tumors of the central nervous system most often involving astrocytes or oligodendrocytes.
They are distinguised from high grade gliomas (III, IV), such as glioblastoma multiforme, by the absence of anaplasia and have a more favorable prognosis.
Surgery alone is not considered curative for LGG and patients are typically treated with adjuvant radiation therapy, chemotherapy, or both.
If a patient relapses, the tumor may be observed to have a higher grade at that time.
Johnson \textit{et al.} studied a cohort of 23 LGG patients who relapsed, many of whom were treated with the chemotherapeutic agent temozolamide (TMZ)~\cite{johnson2014mutational}.
Whole exome sequencing was performed on tumor tissue at diagnosis and at relapse in an effort to characterize the evolution of recurrent glioma.

\begin{figure}
    \begin{subfigure}{0.5\linewidth}
    \centering
    \includegraphics[height=2.2in]{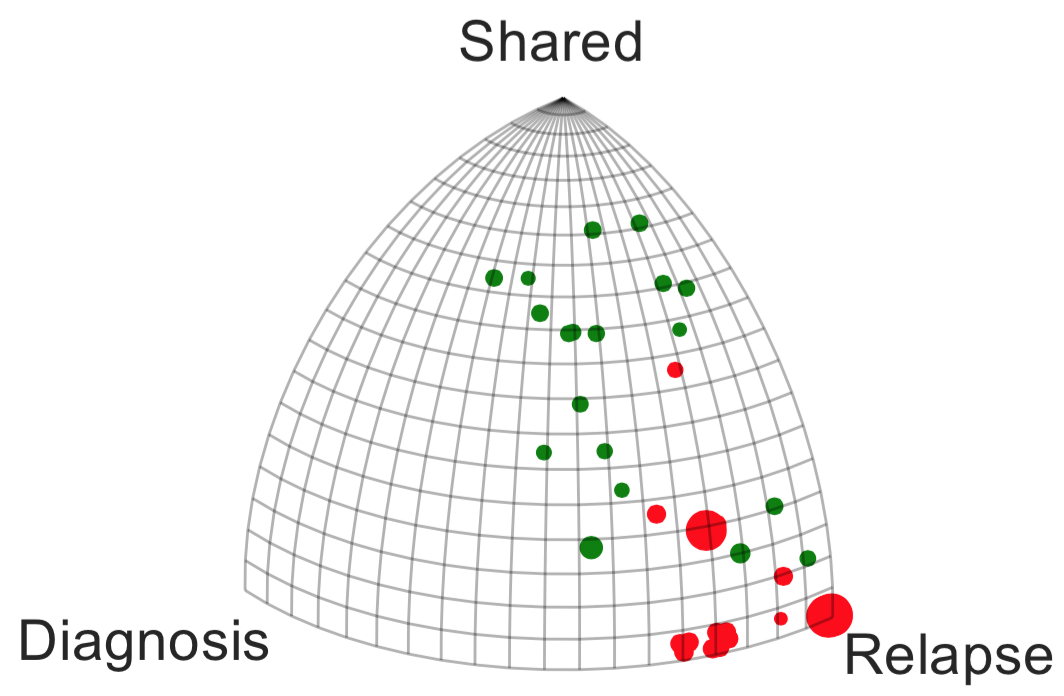}
    \end{subfigure}
    ~
    \begin{subfigure}{0.5\linewidth}
    \centering
    \includegraphics[height=2in]{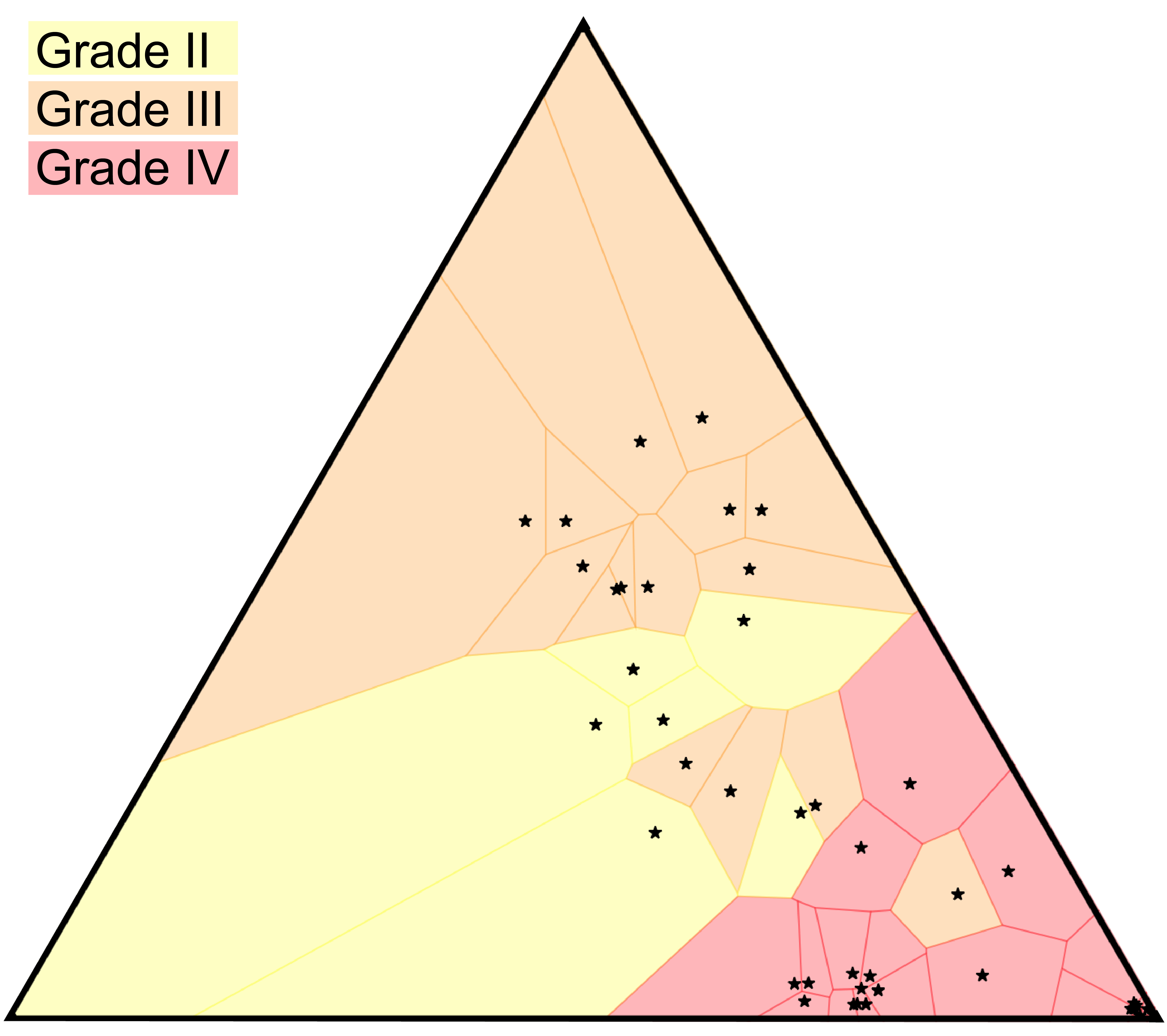}
    \end{subfigure}

    \caption{{\bf Effects of temozolamide (TMZ) treatment in relapsed glioma.} Exome sequencing was performed both at diagnosis and at relapse in 23 glioma patients, allowing for 46 phylogenetic trees to be inferred (spatial replicates in certain cases). (Left) Patients treated with TMZ are colored red, and the size of a point denotes the total number of mutations observed. There is a clear tendency of TMZ-treated patients to localize in a particular corner of the space and to exhibit more mutations associated to the therapy. (Right) Less obvious was the association between the shape of patient's tree and the histologic grade of the tumor at relapse, displayed as a Voronoi tessellation.}
    \label{fig:gliomaTMZ}
\end{figure}

TMZ is an alkylating agent that directly damages the cellular genome, and accordingly we see that the subset of patients that were treated with TMZ show a greater acquisition of relapse-specific mutations.
After projecting the trees to $\mathbb{P}\Sigma_3$, we find the 95\% confidence interval for distance between centroids of the treated / untreated groups is (0.31, 0.48), under 1000-fold bootstrap resampling.
The analogous intervals for untreated / untreated and treated / treated are (0.02, 0.13) and (0.02, 0.16) respectively.
TMZ treatment status defines two statistically well-resolved sub-populations of patients with respect to the shape of their evolutionary behavior, an observation recently reinforced by the larger study of~\cite{wang2016clonal}.
Also of interest is the apparent correlation between the geometry of the phylogenetic tree and the histologic grade at relapse.
A 1-nearest neighbor classifier of this trinary observable (grade II, III, or IV at relapse) yields 85\% accuracy under two-fold cross-validation, and the accuracy does not improve with larger values of $k$.
The tesselation of the space associated to a 1-nearest neighbor classifier is known as a Voronoi diagram, and we have colored the cells in accordance with the grade at relapse.

This example shows that in the simple case of trees with three leaves, evolutionary tumor histories are visibly different under different therapeutic regimes.
Moreover, we see that evolutionary trajectory can be associated with prognosis, as measured by grade at relapse.

\subsection{Clonal dynamics in patient derived xenografts}\label{sec:xeno}

The recent development of single cell transcriptomics and genomics is providing an opportunity to study the role of clonal heterogeneity in tumors~\cite{navin2011tumour, eirew2014dynamics, patel2014single} and to identify small, previously uncharacterized cell populations~\cite{grun2015single}.
The single cell approach to studying complex populations brings with it new challenges associated with the large number of sampled genomes.
Another rapidly maturing technology in the modeling of tumor evolution is that of patient-derived xenografts.
Patient-derived xenografts (PDX) are generated by transplanting tumor tissue into immunodeficient mice, serially engrafting in new mice as each host animal expires.
This provides an \textit{in vivo} platform for drug screening as well as longitudinal monitoring of tumor adaptation and clonal dynamics.

We take advantage of recently published PDX data from breast cancer patients, where single-nucleus deep-sequencing was performed across a lineage of host animals engrafted with a primary lesion of triple-negative breast cancer~\cite{eirew2014dynamics}.
The data are comprised of normal tissue, a sample from the primary tumor, and three subsequent mouse passages.
Somatic mutation calls revealed 55 informative sites of substitution in this lineage, and these variants were assigned to eight distinct cellular populations based on bulk sequencing.

\begin{figure}
    \begin{subfigure}{\linewidth}
    \centering
    \includegraphics[height=4in]{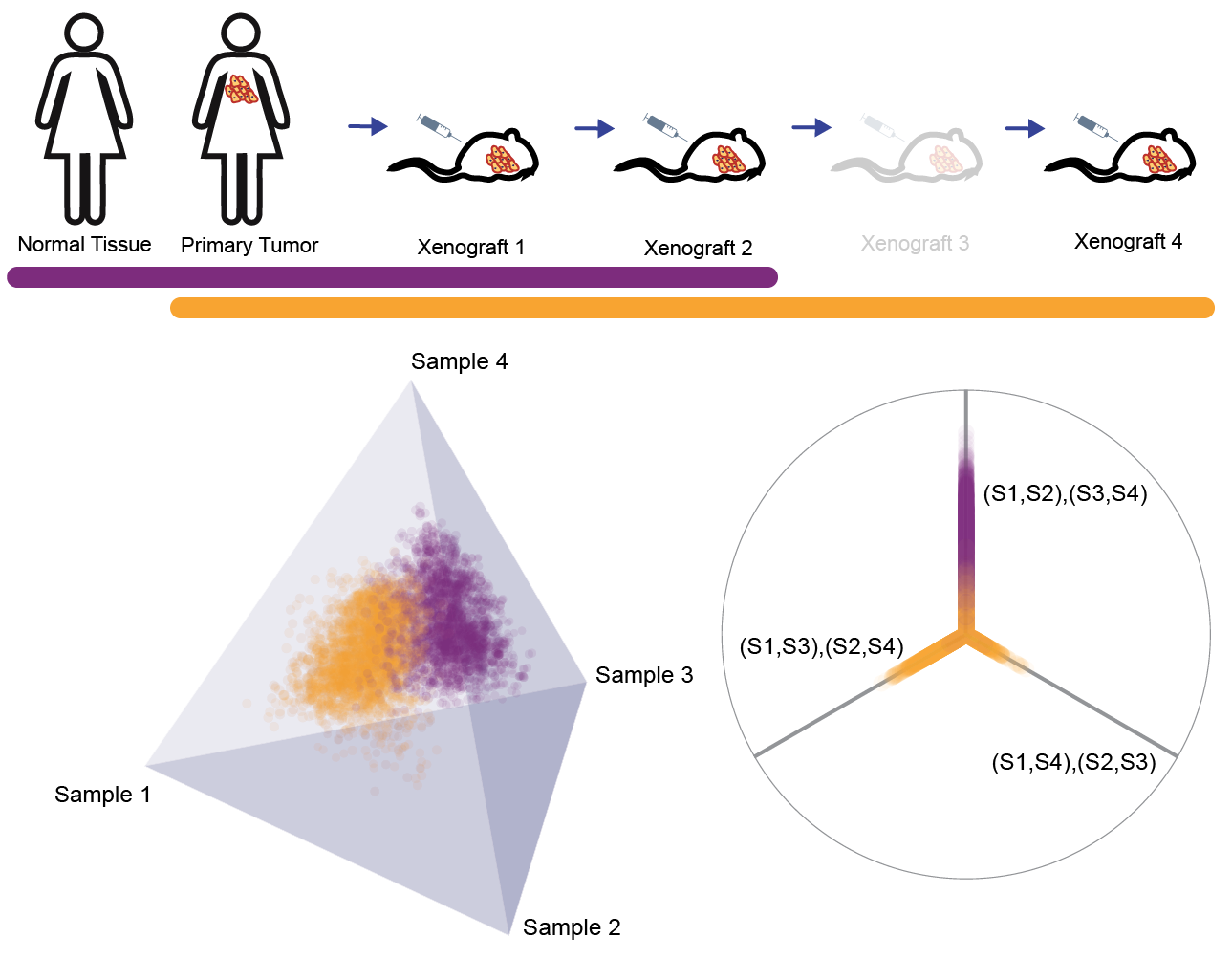}
    \end{subfigure}

    \begin{subfigure}{0.5\linewidth}
    \centering
    \includegraphics[height=2.5in]{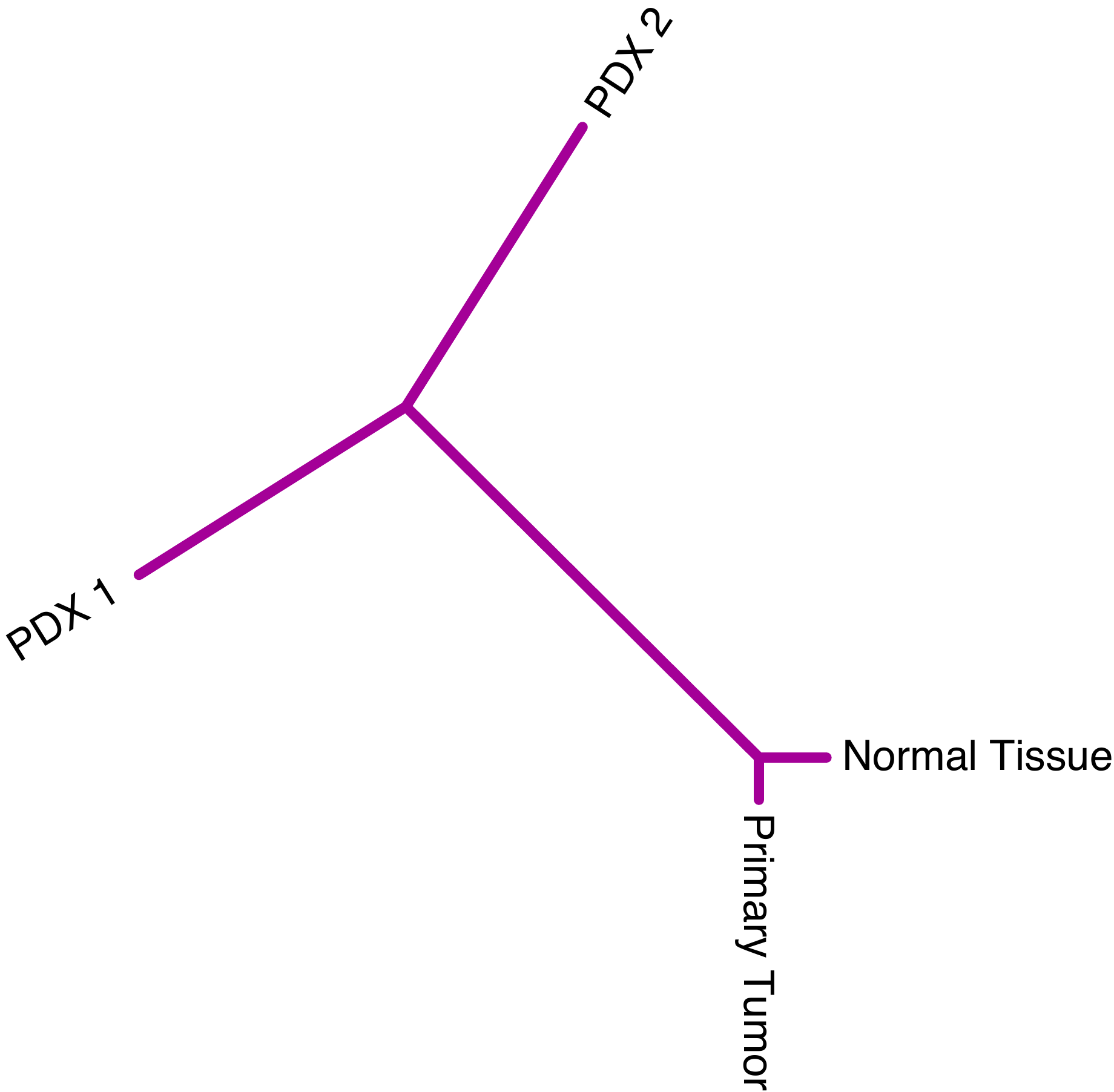}
    \end{subfigure}
    ~
    \begin{subfigure}{0.5\linewidth}
    \centering
    \includegraphics[height=2.5in]{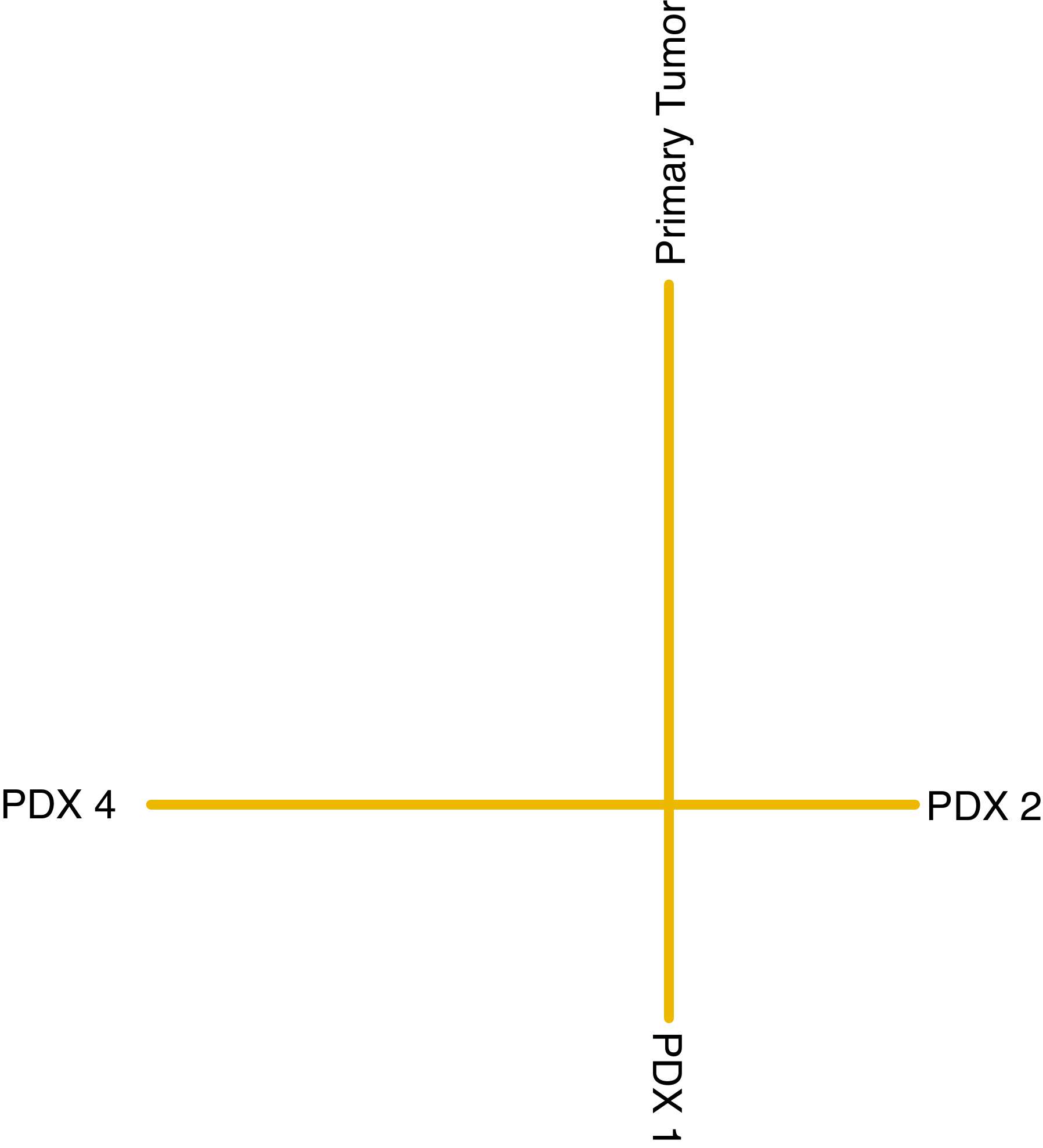}
    \end{subfigure}

    \caption{{\bf Emerging clonal heterogeneity in patient-derived xenograft.} Single cell analysis of tumor evolution in a breast cancer derived xenograft model. Single-nucleus deep-sequence data was obtained from mouse passages 1, 2, and 4, while only bulk sequencing data was available from the primary tumor and matched normal tissue. This data was used to generate two distributions of four-leaved trees, shown in purple and gold in $\mathbb{P}\Sigma_4 $. The former space displays lower standard deviation than the latter, whose centroid is a star tree.}
    \label{fig:xenograft}
\end{figure} 

Only bulk sequencing data were available from the primary tumor and matched normal tissue, while the three xenograft passages were sequenced at single-nucleus resolution.
For each cellular fraction we randomly sampled a single nucleus from the first, second, and fourth PDX passages (27, 36, 27 nuclei respectively were available).
This preparation of the data implies five sequential time points along the tumor's history: benign germline genome, genotype at diagnosis of primary tumor, genotype at first PDX, genotype at second PDX, and genotype at fourth PDX.
The combinatorial possibilities in the latter three time points give rise to a large forest of phylogenetic trees.

We examined the difference in heterogeneity between phylogenetic trees constructed on the first four time points and those constructed on last four time points.
The distributions trees from both time windows are visualized in Figure~\ref{fig:xenograft} as point clouds in $\mathbb{P}\Sigma_4 $.
Consistent linear evolution is seen from primary tumor through the first two xenograft passages, however we observe significant heterogeneity of tumor clones upon the fourth mouse passage.
The first time window (purple) is completely contained within the topology corresponding to linear evolution, unlike the second (gold) which is centered on the origin and extends into all three possible topologies.
The point cloud for the second time window displays a higher standard deviation than the first (10.49 vs. 8.69), and its centroid is essentially a star tree.
Centroids and variances are computed in $\Sigma_4$, prior to rescaling of branch lengths.
The high degree of genotypic heterogeneity giving rise to the second time window distribution is suggestive of a clonal replacement event between the time points of Xenograft 2 (X2) and Xenograft 4 (X4).
Many of the prevalent alterations before X4 disappear during the final passage, and many new mutations rise to dominance.
These results raise interesting questions about the long-term fidelity of PDX vehicles to the genetics of their ancestral primary tumors, which theoretically they serve to mimic.


\section{Detecting new dominant strains in seasonal influenza}\label{sec:flu}

In this section, we will describe an analysis of dynamics in the circulating hemagglutinin sequences of seasonal H3N2 influenza.
Influenza A is an RNA virus that annually infects approximately 5--10\% of adults and 20--30\% of children~\cite{WHO}, leading to more than half a million flu associated deaths.
Vaccination against the virus remains a major way of preventing morbidity.
However, the virus genome evolves rapidly, changing the antigenic presentation of proteins that are in the envelope of the virus, mostly hemagglutinin (HA).
These continuous antigenic changes, often referred as antigenic drift, can lead to failures in vaccine effectiveness.
The design of the influenza vaccine is based on collected isolates of previous years, leaning heavily on the results of hemagglutinin inhibition (HI) assays to detect drift variants.
The great majority of H3N2 isolates are only analyzed antigenically via HI assay, with approximately 10\% of viruses undergoing genetic sequencing of the HA segment~\cite{russell2008influenza}.
Relatively small genetic changes in the genome of the virus can cause drastic antigenic changes.
Influenza vaccine failures can be associated with the emergence of new clones with novel antigenic properties that have replaced recent circulating strains.

A different phenomenon that can lead to significant antigenic changes is reassortment.
The influenza virus genome consists of eight single-stranded RNA segments, two of which code the antigenic surface glycoproteins, hemagglutinin (HA) and neuraminidase (NA).
When two different viruses co-infect the same host, they can generate  progeny containing segments from both parental strains.
Changes in the constellation of segments could introduce dramatic genetic and antigenic changes.
Reassortments could occur between different viruses infecting the same host or even, more rarely, viruses that are typically found in different hosts.
Introducing viral segments from non-human reservoirs has led to major pandemics over the past century~\cite{rabadan2007evolution, rabadan2008non} and was particularly associated with the emergence of the 2009 H1N1 pandemic~\cite{trifonov2009geographic, solovyov2009cluster}.
In 1968, the reassortment of then-circulating H2N2 with avian strains created H3N2 viruses, which have since been infecting human population.
Reassortment of seasonal strains can also contribute to vaccine failure as low frequency hemagglutinin segments could combine with highly transmissible strains.
In particular, reassortment of two H3N2 clades during the 2002-2003 season resulted in a major epidemic and higher incidents of vaccine failure in the succeeding season~\cite{centers2004preliminary}.

In this section, we analyze how the emergence of a novel subclone could be identified by unusual/unexpected tree structures and we develop a genomic predictor of vaccine effectiveness.
We study the recent history of influenza A H3N2 using 1,089 sequences of hemagglutinin collected in the United States between 1993 and 2016 (Figure~\ref{fig:flu_HA_sampling}).
Genomic data was downloaded from the GISAID EpiFlu database (www.gisaid.org), and aligned with MUSCLE~\cite{edgar2004muscle} using default parameters.
Vaccine effectiveness figures were drawn from the meta analysis of~\cite{gupta2006quantifying}.

\begin{figure}
    \begin{subfigure}{\linewidth}
    \centering
    \includegraphics[width=0.8\linewidth]{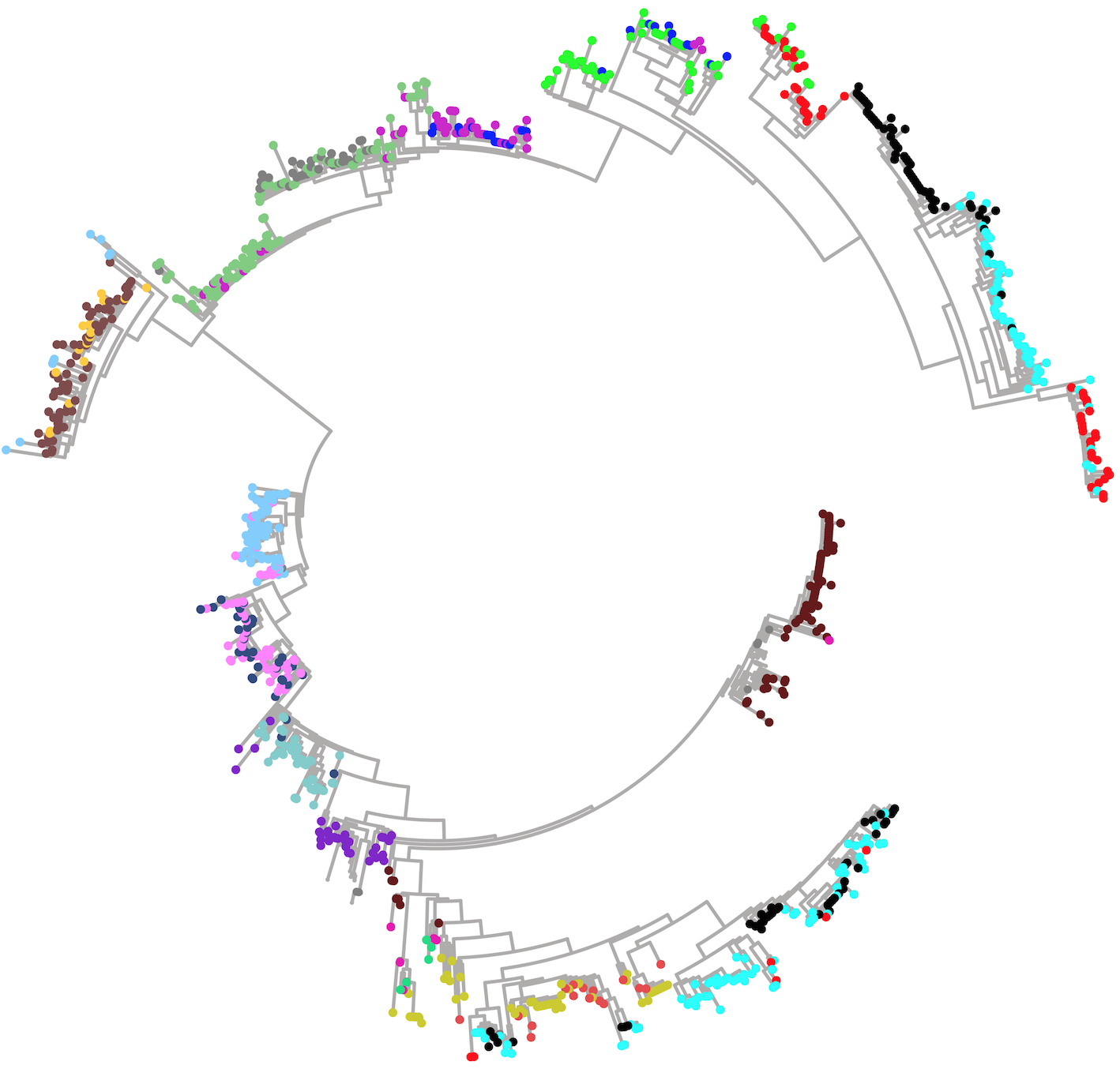}
    \end{subfigure}
    
    \begin{subfigure}{\linewidth}
    \centering
    \includegraphics[width=\linewidth]{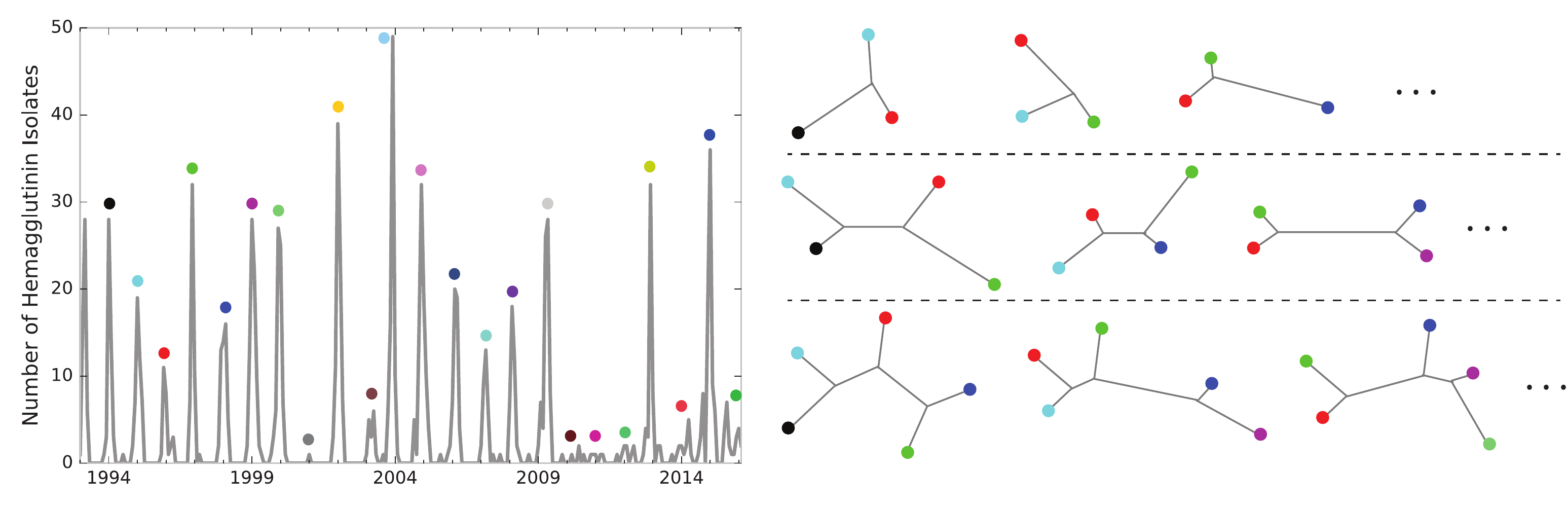}
    \end{subfigure}

    \caption{{\bf H3N2 Hemagglutinin (HA) isolates 1993--2016.} Identifying statistical patterns in large phylogenies is often difficult. (Top) Overall phylogenetic tree inferred from 1,089 sequences, collected in New York state, spanning 24 influenza seasons. (Bottom) There was inter-season variability in the number of H3N2 isolates collected, and we generate sequences of lower dimensional trees by randomly selecting a single HA per season within a temporal window. This procedure decomposes the overall phylogeny into distributions of smaller trees.}
    \label{fig:flu_HA_sampling}
\end{figure}

First, we used tree dimensionality reduction to obtain visualizations of the flu evolutionary profile for exploratory data analysis.
Using 1,089 full length sequences of hemagglutinin collected in New York state between 1993 and 2016, we relate HA sequences from one season to those from the preceding ones.
We randomly select sets of HAs such that a single isolate is drawn from each of three, four, or five consecutive seasons to form a temporal window.
Neighbor-joining with a Hamming metric is used to generate unrooted trees from the temporally ordered tuples of HA isolates.
The case of length five windows is illustrated in Figure~\ref{fig:flu_quint}, where we superimpose each temporal slice onto the same moduli space.
If the current viruses are most similar to viruses circulating in the immediately preceding season, one should expect an unrooted tree topology relating $(1,2),3,(4,5)$ branches.
Deviations from this topology indicate unexpected genetic relationships.
Figure~\ref{fig:flu_quint} confirms that the vast majority of points land along the topologies most compatible with linear evolution of HA.
Certain windows yield well resolved clusters of trees, while others are dispersed point clouds.
Either scenario might be indicative of elevated diversity in the HA segment or a clonal replacement event underway.
The window ending in the 2003-2004 season shows a clear reemergence of strains in 2003-2004 that were genetically similar to those circulating in the 1999-2000 season~\cite{holmes2005whole}.

\begin{figure}
    \centering
    \includegraphics[width=6in]{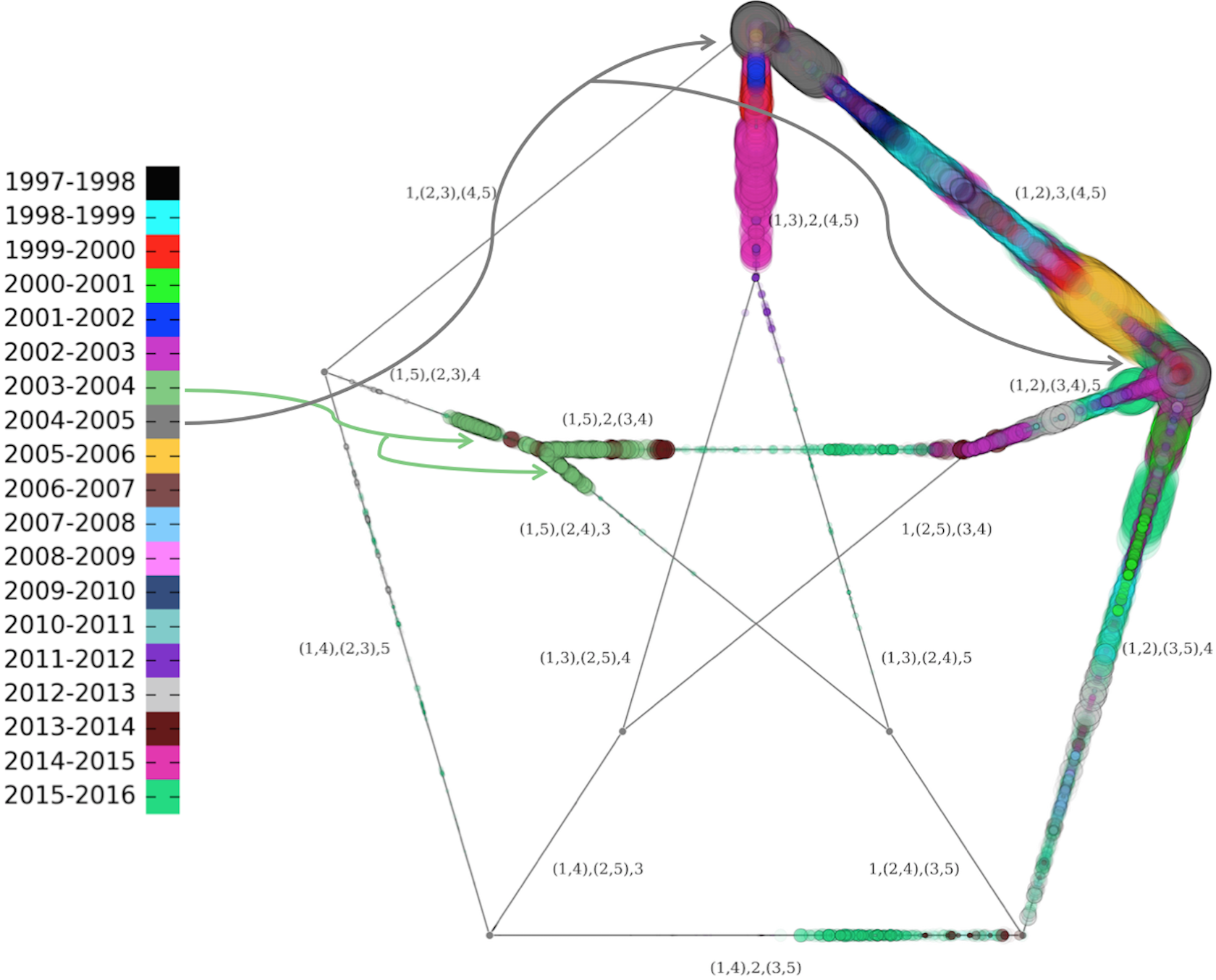}
    \caption{{\bf Temporally windowed subtrees in $\mathbb{P}\Sigma_5$.} Using a common set of axes for projective tree space, we superimpose the distributions of trees derived from windows five seasons long. 1,089 full-length HA segments (H3N2) were collected in New York state from 1993 to 2016. Trees are colored by their most recent season, and point size encodes the magnitude of the cone coordinate. Two consecutive seasons of poor vaccine effectiveness are 2003-2004 and 2004-2005, highlighted with green and gray arrows respectively. The green distribution strongly pairs the 1999-2000 and 2003-2004 strains, hinting at a reemergence.}
    \label{fig:flu_quint}
\end{figure}

\begin{figure}
    \centering
    \includegraphics[width=\linewidth]{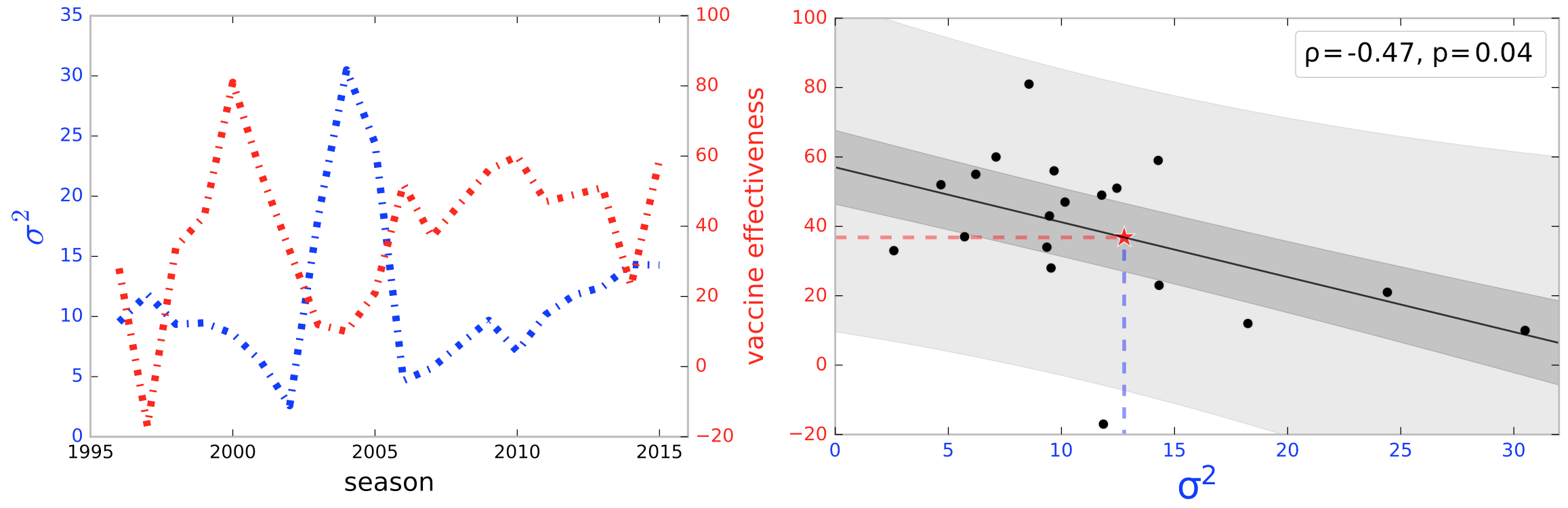}
    \caption{{\bf Diversity in recent circulating HA predicts vaccine failure.} Negative correlation observed between vaccine effectiveness in season $(t, t+1)$ and the variance in trees generated from seasons $(t-1, t), (t-2, t-1), (t-3, t-2)$.}
    \label{fig:flu_VE}
\end{figure}

Next, we used tree dimensionality reduction based on windowing to generate a predictor for vaccine effectiveness.
A natural hypothesis is that elevated HA genetic diversity in circulating influenza predicts poor vaccine performance in the subsequent season.
Distribution features that may intuitively predict future vaccine performance include the variance and the number of clusters in the point cloud.
However, given our limited number of temporal windows, too rich a feature space runs the risk of overfitting, so we focused simply on the variance.
In Figure~\ref{fig:flu_VE} we illustrate the prediction of vaccine effectivenss using the variance of the distribution of trees generated by a lagging window of length 3.
In our notation, a window labeled year $y$ would include the flu season of $(y-1, y)$ and preceding years.
The vaccine effectiveness figures represent season $(y, y+1).$
It is clear, both from the left and right panels, that lower variance in a temporal window predicts increased future vaccine effectiveness, with a Spearman correlation of -0.52 and p-value of 0.02.
The lone outlier season came in 1997-1998~\cite{gupta2006quantifying}, when the vaccine effectiveness was lower than expected.
In this season the dominant circulating strain was A/Sydney/5/97 while the vaccine strain was A/Wuhan/359/95.
The analysis can be carried out with length-4 or length-5 windows to yield a similar result.
Noteworthy is the fact that this association rests only on aligned nucleotide sequence, making no direct use of HA epitope or HI assay data.
The correlation between variance of tree distributions and vaccine effectiveness allows us to estimate the influenza vaccine effectiveness for future seasons based on genomic data.
In particular, for the 2016-2017 season, the variance in tree space was 12.77, corresponding to an approximate effectiveness of 36\%.

We can also use our approach to retrospectively examine the annual W.H.O. decisions to either keep or change the H3N2 componenet of the Northern hemisphere vaccine, and whether the choice resulted in superior or inferior vaccine effectiveness in the following season.
A coarse-grained view of the antigenic features of our HA isolates can be obtained using the work of~\cite{du2012mapping}, who defined a clustering of antigenic phenotype and trained a naive Bayes model, that maps HA protein sequence to these labels.
We begin by labeling our phylogenetic trees using the antigenic cluster (AC) assignments of the classifier, and selecting those seasons in which more than one AC is observed.
The goal is to define a mapping between features of the different AC distributions and the change in vaccine effectiveness of the next season relative to the present.
Figure~\ref{fig:flu_antigenic} indicates that in 9 of the 19 seasons for which we have data, only a single AC was observed.
In this case it does not make sense to ask whether a change in H3N2 vaccine strain should have considered for the subsequent season, since we detect no antigenic diversity.
The other 10 seasons were represented as vectors comprised of 4 features: distance between the centroids of the older / newer AC distributions, standard deviation of the older AC distribution, standard deviation of the newer AC distribution, whether the vaccine strain for the subsequent season was changed from that of the current season.
We associated a binary label to each of the 10 seasons: whether the change in vaccine effectiveness was positive or negative.

Using a heavily restricted vocabulary of logical and arithmetic operators, we exhaustively searched for a decision tree mapping the feature vectors to the binary label, based on a fitness function maximizing area under the receiver operating characteristic~\cite{schmidt2009distilling}.
A decision rule that achieves perfect classification on this data set is depicted in Figure~\ref{fig:flu_antigenic}.
The very small size of the data set means that caution is warranted when interpreting the results.
Nonetheless, the results do suggest that if a vaccine strain is unchanged, a higher variance in the old AC distribution predicts improvement in vaccine effectiveness ($\Delta V.E. > 0$), while if a vaccine strain is changed, then the new AC distribution being well-resolved from the old AC predicts $\Delta V.E. > 0$.

\begin{figure}
    \centering
    \includegraphics[width=\linewidth]{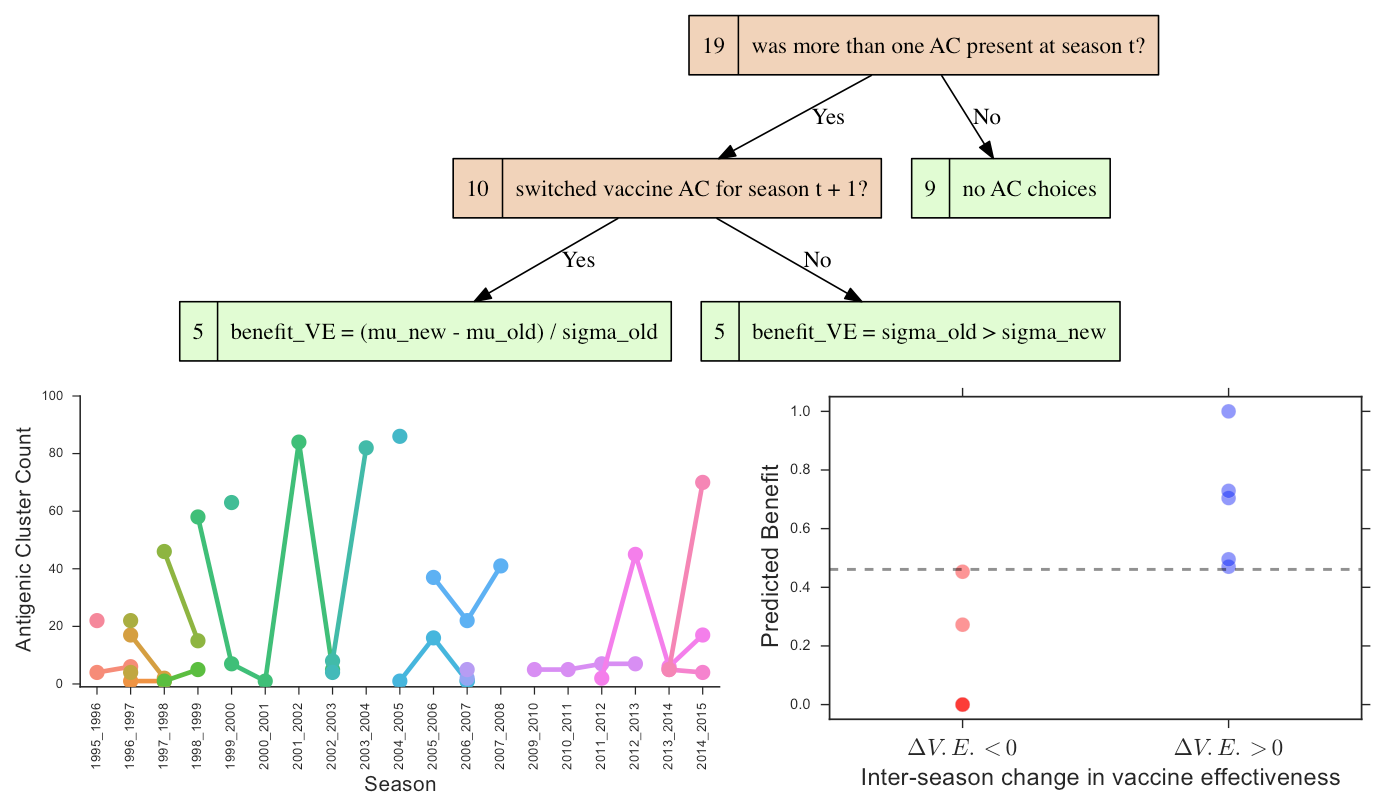}
    \caption{{\bf Stratification of trees on predicted antigenic cluster.} Trees are categorically labeled using the predicted~\cite{du2012mapping} antigenic cluster (AC) of their most recent isolate. (Top) A simple decision rule is fit~\cite{schmidt2009distilling} to the 10 seasons in which more than one AC is observed, explaining the change in vaccine effectivness (VE) in the following season. (Bottom) Colors encode different AC labels associated to the HA sequences collected over time. We also plot the predicted VE change given by the decision rule against the historical data, for the 10 relevant seasons.}
    \label{fig:flu_antigenic}
\end{figure}


\subsection*{Acknowledgments}

The authors gratefully acknowledge the constructive feedback of Gillian Grindstaff, Melissa McGuirl, and Daniel Rosenbloom.
This work was supported by a TL1 personalized medicine fellowship (5TL1TR000082) and NIH grants (R01 CA179044, R01CA185486, R01GM117591, U54 CA193313).

\vspace{0.25in}

\noindent Code associated with this work can be obtained at https://github.com/RabadanLab.


\bibliography{evomod_preprint}


\end{document}